\documentclass[acmsmall,screen,nonacm]{acmart}

\usepackage[english]{babel}

\usepackage{hyphenat}
\usepackage{aliascnt}
\usepackage{amsmath}
\usepackage{faktor}
\usepackage{fancyhdr}
\usepackage{mathtools}
\usepackage{framed}
\usepackage{bbding}
\usepackage{setspace}
\usepackage{stmaryrd}
\usepackage{url}
\usepackage{enumitem}
\usepackage{stackengine}
\stackMath
\usepackage{proof}
\usepackage{listings}
\usepackage{rotating}
\usepackage{natbib}
\usepackage{subcaption}

\graphicspath{{images/}}
\newcommand{\NDOL}{OL$_{\text{nd}}$}
\newcommand{\setN}{\mathbb{N}}
\newcommand{\setZ}{\mathbb{Z}}

\DeclarePairedDelimiter{\denot}{\llbracket}{\rrbracket}
\DeclarePairedDelimiter{\edenot}{\llparenthesis}{\rrparenthesis}
\newcommand{\wlp}{\textbf{wlp}}
\newcommand{\false}{\text{false}}
\newcommand{\true}{\text{true}}

\newcommand{\eqdef}{\triangleq}

\newcommand{\regr}{\mathsf{r}}
\newcommand{\regc}{\mathsf{c}}
\newcommand{\Reg}{\mathsf{RCmd}}
\newcommand{\ACmd}{\mathsf{ACmd}}

\newcommand{\regplus}{\boxplus}
\newcommand{\kstar}{*}
\newcommand{\Var}{\text{Var}}
\newcommand{\AExp}{\mathsf{AExp}}
\newcommand{\BExp}{\mathsf{BExp}}

\makeatletter
\newsavebox{\@brx}
\newcommand{\llangle}[1][]{\savebox{\@brx}{\(\m@th{#1\langle}\)}%
  \mathopen{\copy\@brx\mkern2mu\kern-1\wd\@brx\usebox{\@brx}}}
\newcommand{\rrangle}[1][]{\savebox{\@brx}{\(\m@th{#1\rangle}\)}%
  \mathclose{\copy\@brx\mkern2mu\kern-1\wd\@brx\usebox{\@brx}}}
\makeatother

\newcommand{\underexact}[1]{\ensuremath{[#1]}}
\newcommand{\overexact}[1]{\ensuremath{\{#1\}}}
\newcommand{\necessaryexact}[1]{\ensuremath{(#1)}}
\newcommand{\angleexact}[1]{\ensuremath{\llangle #1 \rrangle}}
\newcommand{\undertriple}[3]{\underexact{#1}~#2~\underexact{#3}}
\newcommand{\overtriple}[3]{\overexact{#1}~#2~\overexact{#3}}
\newcommand{\necctriple}[3]{\necessaryexact{#1}~#2~\necessaryexact{#3}}
\newcommand{\angletriple}[3]{\angleexact{#1}~#2~\angleexact{#3}}
\newcommand{\oltriple}[3]{\langle {#1}\rangle~#2~\langle{#3}\rangle}
\newcommand{\lrule}[1]{\ensuremath{\llangle\mathsf{#1}\rrangle}}
\newcommand{\irule}[1]{[{\ensuremath{\mathsf{#1}}}]}
\newcommand{\horule}[1]{\{{\ensuremath{\mathsf{#1}}}\}}
\newcommand{\olrule}[1]{\ensuremath{\langle\mathsf{#1}\rangle}}

\newcommand{\fwsem}[1]{\denot{#1}}
\newcommand{\bwsem}[1]{\denot{\overleftarrow{#1}}}

\newcommand{\Regh}{\mathsf{HRCmd}}
\newcommand{\Cmdh}{\mathsf{HACmd}}
\newcommand{\Loc}{\text{Loc}}
\newcommand{\Val}{\text{Val}}
\newcommand{\Stores}{\text{Store}}
\newcommand{\Heaps}{\text{Heap}}
\newcommand{\dom}{\text{dom}}
\newcommand{\errstate}{\textbf{err}}
\newcommand{\modified}{\text{mod}}
\newcommand{\dotsim}{\dot\sim}
\newcommand{\Asl}{\text{Asl}}
\newcommand{\emp}{\textbf{emp}}
\newcommand{\andsep}{*}
\newcommand{\asldenot}[1]{\{\hspace{-2pt}|#1|\hspace{-2pt}\}}
\newcommand{\fv}{\text{fv}}
\newcommand{\dealloc}{\not\mapsto{}}

\newcommand{\code}[1]{\texttt{#1}}

\newcommand{\svert}{\,\vert\,}
\newcommand{\sdot}{\,.\,}
\newcommand{\proofcase}[1]{\noindent \textbf{Case} #1 \newline}
\newcommand{\odd}{\mathit{odd}}
\newcommand{\even}{\mathit{even}}

\AtBeginDocument{%
	}
\setcopyright{none}
\acmJournal{PACMPL}

\begin{document}
\newtheorem{remark}[theorem]{Remark}

\title{Sufficient Incorrectness Logic: SIL and Separation SIL}

\author{Flavio Ascari}
\orcid{0000-0003-4624-9752}
\affiliation{\department{Dipartimento di Informatica}\institution{Universit\`{a} di Pisa}\streetaddress{Largo B. Pontecorvo 3}\postcode{56127}\city{Pisa}\country{Italy}}
\email{flavio.ascari@phd.unipi.it}
\author{Roberto Bruni}
\orcid{0000-0002-7771-4154}
\affiliation{\department{Dipartimento di Informatica}\institution{Universit\`{a} di Pisa}\streetaddress{Largo B. Pontecorvo 3}\postcode{56127}\city{Pisa}\country{Italy}}
\email{roberto.bruni@unipi.it}
\author{Roberta Gori}
\orcid{0000-0002-7424-9576}
\affiliation{\department{Dipartimento di Informatica}\institution{Universit\`{a} di Pisa}\streetaddress{Largo B. Pontecorvo 3}\postcode{56127}\city{Pisa}\country{Italy}}
\email{roberta.gori@unipi.it}
\author{Francesco Logozzo}
\affiliation{Meta Platforms, Inc.\country{USA}}
\email{logozzo@meta.com}

\begin{abstract}

Sound over-approximation methods have been proved effective for guaranteeing the absence of errors, but inevitably they produce false alarms that can hamper the programmers.
Conversely, under-approximation methods are aimed at bug finding and are free from false alarms.
We introduce Sufficient Incorrectness Logic~(SIL), a new under-approximating, triple-based program logic to reason about program errors.
SIL is designed to set apart the initial states leading to errors.
We prove that SIL is correct and complete for a minimal set of rules, and we study additional rules that can facilitate program analyses.
We formally compare SIL to existing triple-based program logics. 
Incorrectness Logic and SIL both perform under-approximations, but while the former exposes only true errors, the latter locates the set of initial states that lead to such errors.
Hoare Logic performs over-approximations and as such cannot capture the set of initial states leading to errors in \emph{nondeterministic} programs -- for deterministic and terminating programs, Hoare Logic and SIL coincide.
Finally, we instantiate SIL with Separation Logic formulae (Separation SIL) to handle pointers and dynamic allocation and we prove its correctness and, for loop-free programs, also its completeness.
We argue that in some cases Separation SIL can yield more succinct postconditions and provide stronger guarantees than Incorrectness Separation Logic and can support effective backward reasoning.

\end{abstract}



\maketitle


\section{Introduction}
Formal methods aim to provide tools for reasoning and establishing program guarantees.
Historically, research in formal reasoning progressed from manual correctness proofs to effective, automatic methods that improve program reliability and security.
In the late 60s, \citet{Floyd1967Flowcharts} and \citet{hoare69} independently introduced formal systems to reason about programs.
In the 70s/early 80s, the focus was on mechanization, with the introduction of numerous techniques such as predicate transformers \citep{DBLP:journals/cacm/Dijkstra75}, Abstract Interpretation \citep{DBLP:conf/popl/CousotC77}, model checking \citep{DBLP:conf/lop/ClarkeE81}, type inference \citep{DBLP:conf/popl/DamasM82} and mechanized program proofs~\citep{DBLP:conf/eurocal/CoquandH85}.
Those seminal works, in conjunction with the development of automatic and semi-automatic theorem provers (e.g.,~\cite{DBLP:conf/cade/Moura07}) brought impressive wins in proving program correctness of real-world applications.
For instance, the Astrée abstract interpreter automatically proves the absence of runtime errors in millions of lines of safety-critical C~\cite{DBLP:conf/pldi/BlanchetCCFMMMR03}, the SLAM model checker was used to check Windows drivers~\cite{DBLP:conf/cav/BallR01}, CompCert is a certified C compiler developed in Coq~\cite{DBLP:journals/cacm/Leroy09}, and VCC uses the calculus of weakest precondition to verify safety properties of annotated Concurrent C programs~\cite{DBLP:conf/tphol/CohenDHLMSST09}.

Despite the aforementioned successes, effective program correctness methods struggle to reach mainstream adoption.
As program correctness is undecidable, all those methods \emph{over-approximate} programs behaviours.
Over-approximation guarantees soundness (if the program is proved to be correct, no error will arise) but it may introduce many false positives.
False positives are seen as a distraction by professional programmers, who are skeptical of the adoption of tools with a high false positive ratio.
Tools such as Microsoft Prefix~\citep{DBLP:conf/icse/HackettDWY06} or Microsoft CodeContracts~\citep{DBLP:conf/foveoos/FahndrichL10} aim at reducing the number of false positives by using human-provided annotations. 
Nevertheless, in some scenarios, the annotation effort may be over-killing so that even those methods struggle to enjoy universal adoption.

To address the aforementioned issues, researchers have recently invested more efforts in effective \emph{under-approximations}.
The overall goal is to develop \emph{principled} techniques to assist programmers with timely feedback on their code about the presence of true errors, with few or zero false alarms.
For instance, Microsoft Sage combines SMT solving with program executions (which are an under-approximation of the program semantics) to find security bugs~\citep{DBLP:journals/cacm/GodefroidLM12} and the Bounded Model Checking technique unrolls recursion up to a certain depth to find property violations~\citep{DBLP:journals/ac/BiereCCSZ03}.

Recently, \citet{DBLP:journals/pacmpl/OHearn20} introduced Incorrectness Logic (IL).
Intuitively, IL uses under\hyp{}approximations to reduce the problem complexity by dropping disjunctions during the state space exploration.
IL is the theoretical foundation of many successful bug-catching tools \citep{DBLP:conf/cav/RaadBDDOV20,DBLP:journals/pacmpl/RaadBDO22, DBLP:journals/cacm/DistefanoFLO19,DBLP:journals/pacmpl/LeRVBDO22,DBLP:conf/cpp/CarbonneauxZKON22} and the first Workshop on ``Formal Methods for Incorrectness'' affiliated with POPL'24 witnesses the interest around this new line of research.
IL triples resemble Floyd-Hoare logic ones, but they have a different meaning: any error state that satisfies the postcondition can be produced by some input states that satisfy the precondition. Since in this paper we deal with different kinds of program logics, a compact legenda of the different notation used in the literature for the various triples is reported in Fig.~\ref{fig:legenda}.
It is worth noting that in IL nothing is said about \emph{which input states are responsible for a given error}.
This is possibly rooted in the \emph{forward} flavour of the under-approximation, which follows the ordinary direction of code execution.

\begin{figure}[t]
	\centering
	\begin{tabular}{c|@{\qquad}c@{\qquad}c@{\qquad}c@{\qquad}c@{\qquad}c@{\qquad}}
	Logic & HL & IL & OL & NC & SIL 
	\\[4pt] \hline &&&&& \\[-6pt]
	Triples & 
	$\overtriple{P}{\regr}{Q}$ &
	$\undertriple{P}{\regr}{Q}$ &
	$\oltriple{P}{\regr}{Q}$ &
	$\necctriple{P}{\regr}{Q}$ &
	$\angletriple{P}{\regr}{Q}$ 
	\end{tabular}
	\caption{Legenda: the figure shows the different syntax for the various triples used in this paper.}
	\label{fig:legenda}
\end{figure}

\subsection{Backward under-approximation and Lisbon Triples}\label{sec:origin}
At POPL'19 in Lisbon, according to the origins of IL first narrated by~\citet[\S 7]{DBLP:journals/pacmpl/OHearn20} and further discussed in~\citet[\S 8]{DBLP:journals/corr/ZilbersteinDS23}, Derek Dreyer and Ralf Jung suggested Peter O'Hearn to look at bug-finding in terms of a logic for proving the presence of faults. However, the proposed model of triples did not fit well with a key feature of Pulse, a bug-catching tool developed mainly by Jules Villard at Facebook, namely its ability to drop the analysis of some program paths, for which IL provides a sound logical foundation instead. The idea of such `Lisbon' triples is that \emph{for any initial state satisfying the precondition, there exists some trace of execution leading to a final state satisfying the postcondition} and it can be dated back to Hoare’s calculus of possible correctness~\cite{DBLP:journals/jacm/Hoare78}, even if no form of approximation was considered there. This condition aims at finding input states causing bugs: if the postcondition denotes an error, any initial state satisfying the precondition exhibits at least one unwanted behaviour. Lisbon triples were then briefly accounted for in~\citet[\S 5]{DBLP:conf/RelMiCS/MollerOH21} and \citet[\S 3.2]{DBLP:journals/pacmpl/LeRVBDO22}, under the name \emph{backwards under-approximate triples}. Later, the idea of finding sources of error was one of the motivations leading to the more general Outcome Logic (OL)~\cite{DBLP:journals/corr/ZilbersteinDS23}, which is able to reason about multiple executions at the same time via the outcome conjunction~$\oplus$. \citet{DBLP:journals/corr/ZilbersteinDS23} encode Lisbon triples as OL triples matching the specific pattern $\oltriple{P}{\regr}{Q\oplus \top}$. OL can also exploit triples \`a la HL for partial correctness. In fact, OL has been designed to provide a single unified theory that can be used for both correctness and incorrectness reasoning.
A sound and complete proof system for Lisbon triples has been independently proposed in the recent preprint \citep{unter2024}, where they are combined with IL triples for non-termination analysis of programs.

Ideally, given an incorrectness specification, the goal of backward under-approximation would be to report their \emph{weakest preconditions} to the programmer, for discerning all dangerous input states that lead to bugs.
However, all the above proposals exploit logic rules that are designed according to the forward semantics of programs and thus are best suited to infer postconditions starting from some given precondition: in general, their backward analysis may require some ingenious guess of the right predicates to instantiate the consequence rule with.
Here, by exploiting backward program semantics, we design a minimal, sound and complete proof system for Lisbon triples that is aimed at inferring weakest preconditions.

\subsection{Sufficient Incorrectness Logic}
\begin{figure}
\begin{subfigure}[t]{0.48\textwidth}
\begin{lstlisting}[language=C]
// program r42

if (x is even) {
    if (y is odd) {
        z := 42;
    }
}
// assert(z != 42)
\end{lstlisting}
\caption{A deterministic program where SIL is different from IL.}
\label{fig:simple_example}
\end{subfigure}
\hfill
\begin{subfigure}[t]{0.48\textwidth}
\begin{lstlisting}[language=C]
// program r42nd
x := nondet();
if (x is even) {
    if (y is odd) { 
        z := 42;
    }
}
// assert(z != 42)
\end{lstlisting}
\caption{A nondeterministic program where SIL is different from HL.}\label{fig:simple_example_nondet}
\end{subfigure}
\caption{For those two programs, SIL can prove that $\even(x) \land \odd(y)$ is a sufficient precondition for the assertion violation.}
\end{figure}

In this paper, we introduce Sufficient Incorrectness Logic (SIL), a backward\hyp{}oriented under\hyp{}approximation proof system that focuses on \emph{finding the sources of incorrectness rather than just highlighting the presence of bugs}.
SIL formalizes exactly Lisbon triples: a valid SIL triple
\[
\angletriple{P}{\regr}{Q} 
\]
can be interpreted as ``\emph{all input states that satisfy $P$ have at least one execution of the program $\regr$ leading to a state that satisfies $Q$}".
In the absence of nondeterminism, SIL guarantees that a state satisfying the precondition \emph{always} leads to an error.
In the presence of nondeterminism, SIL guarantees that there \emph{exists} an execution that leads to an error.
Sufficient incorrectness preconditions are extremely valuable to programmers, because, by pointing out the sources of errors, they serve as starting point to scope down debugging, fuzzing, and testing in general, as was already observed in \citet[§2.3]{DBLP:journals/corr/ZilbersteinDS23}.
Differently from OL, SIL goals include: (i)~defining a default deduction mechanism that starts from some specification of erroneous outcomes and traces the computation back to some initial states that can produce such errors; (ii)~exhibiting a minimal set of rules that are complete and correct for backward under-approximation; and (iii)~spelling out, a posteriori, a formalization and generalization of the backward analysis step performed by industrial grade analysis tools for security developed at Meta, such as Zoncolan~\citep{DBLP:journals/cacm/DistefanoFLO19}, Mariana Trench~\citep{MarianaTrench}, and Pysa~\citep{Pysa} (see Section~\ref{sec:tool}).

\subsection{Why Sufficient Incorrectness Logic?}
SIL characterizes the initial states leading to an error.
However, one may ask if we do need to develop a new logic and whether existing ones suffice.
This is a legitimate question, and we answer it by comparing SIL with existing program logics.

\subsubsection{SIL and Incorrectness Logic}\label{sec:SILvsIL}
Let us consider the program $\regr 42$ in Figure~\ref{fig:simple_example}, also used by \citet{DBLP:journals/pacmpl/OHearn20} to show the essence of IL triples.
We assume that $Q_{42} \eqdef (z = 42)$ denotes the set of erroneous states.
Any SIL triple provides a sufficient condition for the input state that causes the error. 
For instance, with the deduction system introduced in this paper (but not with IL~\citep[§3.1]{DBLP:journals/pacmpl/OHearn20}) we can prove the triple
\[
\angletriple{\odd(y) \wedge \even(x)}{\regr 42}{Q_{42}}
\]
stating that any state where $x$ is even and $y$ is odd will lead to the error (i.e., $z=42$).
In IL, but \emph{not} in SIL, one can prove that some errors can be reached even when we start in a safe state.
For instance, the following triple holds in IL:
\begin{equation}
\label{eq:il-simple-det}
\undertriple{z=11}{\regr 42}{Q_{42} \wedge \odd(y) \wedge \even(x)}.
\end{equation}
However, that triple is not valid in SIL, because it is not true that for any state where $z=11$ the program will reach an error, e.g., when we start from the state $z\mapsto 11,x\mapsto 0,y\mapsto 0$.

\subsubsection{SIL and Hoare logic}\label{sec:sil-and-hl}
If a program $\regr$ is deterministic and non-divergent, then SIL triples are equivalent to Hoare triples---we only interpret the postcondition in the latter as the error states.
For instance, in the example of Figure~\ref{fig:simple_example}, the program $\regr 42$ is trivially deterministic and non-divergent, so the Hoare triple 
\[
\overtriple{\odd(y) \wedge \even(x)}{\regr 42}{Q_{42}}
\]
is valid (where $Q_{42} \eqdef (z=42)$ as before).
However, this is no longer true in the presence of divergence and nondeterminism.
To illustrate this, consider the program $\mathsf{r42nd}$ in Figure~\ref{fig:simple_example_nondet}, which introduces nondeterminism in the program $\regr 42$ of Figure~\ref{fig:simple_example}.
The SIL triple 
\begin{equation}
\label{eq:sil:nondet}
\angletriple{\odd(y) \wedge \even(x)}{\mathsf{r42nd}}{Q_{42}}
\end{equation}
is still valid: any state satisfying the precondition can lead to the error when an even value is assigned to $x$.
However, the Hoare triple 
\[
\overtriple{\odd(y) \wedge \even(x)}{\mathsf{r42nd}}{Q_{42}}
\] 
is no longer valid, because starting from an initial state where $\mathtt{z} \neq 42$, $\mathtt{x}$ may be assigned an odd value, hence never entering the conditional.
Otherwise stated, in Hoare triples, the postcondition must over-approximate all the possible outcomes, which in the presence of nondeterminism (or divergence) means that it cannot single out the error states, as under-approximating logics can do.

\subsubsection{SIL and Outcome logic}\label{sec:sil-and-ol}
Outcome Logic (OL)~\citep{DBLP:journals/corr/ZilbersteinDS23} extends HL by allowing to express both correctness and incorrectness properties. Although OL has been designed as a general framework that is parametric in a monoid of outcomes, the instance we compare with SIL is based on the powerset monoid, called the nondeterministic instance of OL (\NDOL).
Even though Lisbon triples can be expressed in OL, we show a simple example of a Lisbon triple whose proof in SIL is straightforward but which we were not able to derive in \NDOL.
To understand the example, it is worth noting that in \NDOL{} there are three different forms of disjunctive predicates that coexist with different meanings. To see this, take the two atomic assertions $P_x \eqdef (x = 0)$ and $P_y \eqdef (y = 0)$. A set of states $m$ satisfies their union $P_x \cup P_y$ if $\forall \sigma\in m. (\sigma(x)=0 \vee \sigma(y)=0)$. 
Instead, $m$ satisfies the disjunction $P_x \vee P_y$ if $(\forall \sigma\in m. \sigma(x)=0) \vee (\forall \sigma\in m. \sigma(y)=0)$. Lastly, $m$ satisfies the outcome composition $P_x \oplus P_y$ if $m$ can be decomposed as the union of two non-empty sets $m_x$ and $m_y$ that satisfy $P_x$ and $P_y$ respectively. The difference between $P_x \cup P_y$ and $P_x \vee P_y$ should be clear. If we take $m \eqdef\{ \sigma \mid \sigma(x)=0 \wedge \sigma(y)=1\}$, then $m$ satisfies both $P_x \cup P_y$ and $P_x \vee P_y$, but not $P_x \oplus P_y$.
 
Now take the nondeterministic code
\begin{equation}
	\regr_{xy} \eqdef \code{((y=0)?; x := 0)}\; \regplus\; \code{((x=0)?; y := 0)}\label{eq:sil-and-ol-example-triple}
\end{equation}
and let $Q\eqdef (x=0 \cap y=0)$ be an incorrectness specification. By straightforward application of SIL rules we can derive the Lisbon triple:
\[
 \angletriple{P_x\cup P_y}{\regr_{xy}}{Q}
\]
The corresponding \NDOL{} triple is
\[
 \oltriple{P_x\cup P_y}{\regr_{xy}}{Q\oplus \top}
\]
but, by using the rules available in~\citet{DBLP:journals/corr/ZilbersteinDS23}, to our best efforts, we were only able to prove the \NDOL{} triple:
\[
 \oltriple{P_x\oplus P_y}{\regr_{xy}}{Q\oplus \top} .
\]

Notably, the simpler Lisbon triples $\angletriple{P_x}{\regr_{xy}}{Q}$ and $\angletriple{P_y}{\regr_{xy}}{Q}$ can be derived in both SIL and \NDOL{} proof systems, but the single triple $\angletriple{P_x\cup P_y}{\regr_{xy}}{Q}$ exposes the weakest SIL precondition for the specification $Q$, which is useful to determine the least general assumptions that may lead to an error. Thus the example suggests that weakest preconditions are not necessarily derivable in \NDOL.
Note that a revised and improved version of OL is currently under development \citep{zilberstein2024relatively}, where completeness is achieved by introducing, among others, a dedicated rule for disjunction and exploiting the equivalence $P_x\cup P_y \equiv P_x \vee P_y \vee (P_x\oplus P_y)$. 

\subsubsection{SIL and Necessary Preconditions}\label{sec:introSILvsNC}
\citet{DBLP:conf/vmcai/CousotCL11,DBLP:conf/vmcai/CousotCFL13} introduced Necessary Conditions~(NC) as a principled foundation for automatic precondition inference.
Intuitively a necessary precondition is such that it removes entry states that inevitably will lead to errors \emph{without} removing any good execution.
Therefore, whereas sufficient conditions (e.g., weakest liberal preconditions) require the caller of a function to supply parameters that will never cause an error, necessary conditions only prevent the invocation of the function with arguments that will definitely lead to some error.
Whilst Necessary Conditions and SIL have overlaps (e.g., they focus on providing conditions for errors) they express two different concepts.

Let us illustrate it with the example of Figure~\ref{fig:simple_example_nondet} with again error states $Q_{42} \eqdef (z=42)$.
The correctness specification is thus $\neg Q_{42}$, that is $z \neq 42$.
A necessary precondition for correctness is also $z \neq 42$ because, no matter which nondeterministic value gets assigned to $\mathtt{x}$, if $z = 42$ on entry, the program will reach an error state.
We denote it with the NC triple 
\begin{equation}
    \label{eq:nondet:nc}
\necctriple{z\neq 42}{\mathsf{r42nd}}{\neg Q_{42}}.
\end{equation}
Please note that the precondition of the SIL triple~(\ref{eq:sil:nondet}) has a non-empty intersection with the necessary condition in~(\ref{eq:nondet:nc}), i.e., the formula $z\neq 42 \wedge \odd(y) \wedge \even(x)$ is satisfiable.
This should not come as a surprise: due to nondeterminism, a state satisfying that formula admits both correct and erroneous executions.

\subsection{A Taxonomy of Triple-Based Program Logics}\label{sec:taxonomy}
\begin{figure}[t]
	\centering
	\begin{tabular}{c|ccc}
		& Forward & & Backward\\[3pt]
		\hline &&& \\[-5pt]
		\quad Over\quad\  &\quad Hoare Logic  & $\xleftrightarrow{\qquad\simeq\qquad}$ & Necessary Conditions \\ [5pt]
		\quad Under\quad\  & \quad Incorrectness Logic & & \textcolor{ACMBlue}{Sufficient Incorrectness Logic}
	\end{tabular}
	\caption{The taxonomy of triple-based logics. The column determines whether it is based on forward or backward semantics. The row determines whether it performs under- or over-approximation. HL and NC are equivalent ($\simeq$). SIL is the new logic we introduce here.}
	\label{fig:square-full}
\end{figure}

Our initial motivation for the work in this paper emerged from the attempt to characterize the validity conditions for existing triple-based program logics as under- or over-approximations of forward or backward semantics. 
This led to the taxonomy of HL, IL and NC depicted in Figure~\ref{fig:square-full}.
In this formalization process, we realized there was a missing combination, the one that originated SIL.

The taxonomy is obtained as follows.
Given a (possibly nondeterministic) program $\regr$ let us denote by $\fwsem{\regr}$ (resp. $\bwsem{\regr}$) its forward (resp. backward) collecting semantics.
The forward semantics $\fwsem{\regr}P$ denotes the set of all possible output states of $\regr$ when execution starts from a state in $P$ (and $\regr$ terminates).
Vice versa, the backward semantics $\bwsem{\regr}Q$ denotes the largest set of input states that can lead to a state in $Q$.

We define the validity of the various logics in terms of forward and backward semantics.

\begin{definition}[Triples validity]
\label{def:validity}
For any program $\regr$ and sets of states $P,Q$ we let the following:

\smallskip

\begin{tabular}{llll}
\textbf{HL triples:} & $\overtriple{P}{\regr}{Q}$ & is valid if & $\fwsem{\regr} P \subseteq Q$; \\
\textbf{IL triples:} & $\undertriple{P}{\regr}{Q}$ & is valid if & $\fwsem{\regr} P \supseteq Q$; \\
\textbf{NC triples:} & $\necctriple{P}{\regr}{Q}$ & is valid if & $\bwsem{\regr}Q \subseteq P$; \\
\textbf{SIL triples:} & $\angletriple{P}{\regr}{Q}$ & is valid if & $\bwsem{\regr}Q \supseteq P$.
\end{tabular} 
\end{definition}

By using the validity conditions in Definition~\ref{def:validity}, we can characterize HL, IL, NC, and SIL according to (i) whether the condition is expressed in terms of forward or backward semantics, and (ii) whether it is an over- or an under-approximation (Figure~\ref{fig:square-full}).
Indeed, in HL and NC it is safe to enlarge the respective target set (i.e., $Q$ and $P$) so we do categorize them as over-approximations.
Similarly, in IL and SIL it is safe to shrink their target sets, so we do classify them as under-approximations.
Furthermore, from the validity conditions, we can immediately derive the consequence rules for each logic: for HL and SIL we are allowed to weaken $Q$ and strengthen $P$, while for IL and NC we can do the opposite, i.e. the direction of the consequence rules are determined by the diagonals in Figure~\ref{fig:square-full}. 
Notwithstanding the fact that NC and IL share the same consequence rule, we will show that NC and IL are not related, whereas NC is tightly connected to HL. 
We obtain that an NC triple $\necctriple{P}{\regr}{Q}$ is valid \emph{if and only if} the HL triple $\overtriple{\lnot P}{\regr}{\lnot Q}$ is valid. 
By duality, one might expect a similar connection to hold between IL and SIL, but we show this is not the case.

Interestingly, IL and SIL share the possibility to drop disjuncts, which was in fact one of the leading motivations for the introduction of IL to increase scalability and make the methods effective: the difference is that IL drops disjuncts in the postconditions, while SIL in the preconditions, a feature that is illustrated in Section~\ref{sec:separation-sil-derivation}.

\subsection{Separation Sufficient Incorrectness Logic}\label{sec:introSSIL}
\begin{figure}[t]
\begin{minipage}[t]{0.4\textwidth}
\begin{lstlisting}[language=C]
// program rclient
x := *v;
push_back(v);
*x := 1;
\end{lstlisting}
\end{minipage}\qquad
\begin{minipage}[t]{0.4\textwidth}
\begin{lstlisting}[language=C]
push_back(v) {
  if (nondet()) {
    free(*v);
    *v := alloc();
  }
}
\end{lstlisting}
\end{minipage}
\caption{An example (from \citet{DBLP:conf/cav/RaadBDDOV20}) where Separation SIL produces a more compact postcondition than Incorrectness Separation Logic.}
\label{fig:separationexample}
\end{figure}

Using classical atomic commands for heap manipulation and a suitable language of assertions, we instantiate SIL to the case of Separation Logic, to derive a novel sound and (relatively) complete logic that can identify the causes of memory errors.
We call that new logic Separation SIL.
Intuitively, Separation SIL borrows the ability to deal with pointers from Separation Logic and combines it with the backward under-approximation principles of SIL.

We exemplify Separation SIL using the motivating example of~\citet{DBLP:conf/cav/RaadBDDOV20} for Incorrectness Separation Logic (ISL).
Consider the program in Figure~\ref{fig:separationexample}.
The authors derive the incorrectness triple
\begin{equation}
\label{eq:separation:example:il}
\undertriple{v \mapsto z \andsep z \mapsto -}{\mathsf{rclient}}{v \mapsto y \andsep y \mapsto - \andsep x \dealloc}
\end{equation}
Using Separation SIL we can derive the triple
\begin{equation}
\label{eq:separation:example:separationsil}
\angletriple{v \mapsto z \andsep z \mapsto - \andsep \true}{\mathsf{rclient}}{x \dealloc \andsep \true}
\end{equation}

The triple~\eqref{eq:separation:example:il} proves the existence of a faulty execution starting from at least one state in the precondition.
However, the Separation SIL triple~\eqref{eq:separation:example:separationsil} has both a more succinct postcondition capturing the error and a stronger guarantee: \emph{every} state in the precondition reaches the error, giving (many) actual witnesses for testing and debugging the code.
In Section~\ref{sec:separation-sil-derivation} we show the derivation and in Section~\ref{sec:separation-sil-principles} comment on how SIL principles apply in it.

\subsection{Contributions}
After comparing our work with well-established results in Section~\ref{sec:relatedwork} and recalling some basic concepts and notation in Section~\ref{sec:background}, we make the following three main contributions:
\begin{itemize}
	\item In Section~\ref{sec:sil}, we introduce SIL, a novel logic for sufficient incorrectness triples. SIL enhances program analysis frameworks with the ability to identify the source of incorrectness. We prove that SIL is correct and complete for a minimal set of rules, and we propose some additional rules for program analysis.
	\item Section~\ref{sec:comparison} gives some insights into the relations among SIL, HL, IL and NC by constructing a taxonomy of program logics. Following the classification, we prove that IL and SIL perform under\hyp{}approximations in different ways, i.e., finding the existence vs. the causes of errors, that SIL and HL coincide for deterministic and terminating programs, and that NC is isomorphic to HL and is not related to IL.
	\item In Section~\ref{sec:separation-sil}, we instantiate SIL to Separation Logic formulae to handle memory errors. We prove that Separation SIL is correct and, for loop\hyp{}free programs, is also complete. By means of a simple single example, we claim that Separation SIL triples can be more convenient than ISL ones for testing and debugging and that our backward-oriented deduction strategy is more natural than OL's.
\end{itemize}

We conclude in Section~\ref{sec:conc} pointing out some future work. Appendix~\ref{sec:proofs} contains proofs and minor technical results.

\section{Related work}\label{sec:relatedwork}
The origins of triples for highlighting sufficient preconditions for incorrectness have already been discussed in Section~\ref{sec:origin}.
In the literature, backwards under\hyp{}approximation triples were mentioned in \citet{DBLP:conf/RelMiCS/MollerOH21,DBLP:journals/pacmpl/LeRVBDO22}, but neither of them fully developed a corresponding program logic. Instead, we develop such a backward\hyp{}oriented, \emph{sound and complete} proof system for the first time, study the properties of the deductive system, compare it with existing program logics, and instantiate it for Separation Logic.
Moreover, \citet[§3.2]{DBLP:journals/pacmpl/LeRVBDO22} introduce the notion of \emph{manifest errors} in the context of IL, but mention that backwards under-approximation is able to characterize them directly. Intuitively, an error is manifest if it happens regardless of the context. Manifest errors can be easily captured in SIL: thanks to the completeness result, a postcondition $Q$ defines a manifest error if and only if the SIL triple $\angletriple{\true}{\regr}{Q}$ is provable.

The idea of tracking the sources of errors was one of the motivations that led to Outcome Logic~\citep{DBLP:journals/corr/ZilbersteinDS23}, which is able to express both SIL and HL triples. Manifest errors can also be characterized with OL triples of the form $\oltriple{\true}{\regr}{Q \oplus \top}$. Some differences between OL and SIL have been already sketched in Section~\ref{sec:sil-and-ol} and we refer to Section~\ref{sec:silvsol} for a more technical comparison. In particular, we argue that SIL proof rules are simpler to use for inferring the source of errors starting from an incorrectness specification.

Exact Separation Logic (ESL)~\citep{DBLP:conf/ecoop/MaksimovicCLSG23} combines Separation Logic and Incorrectness Separation Logic to recover the exact semantics of a program and produce function summaries that can be used for both correctness and incorrectness. In the schema in Figure~\ref{fig:square-full}, ESL would be placed between HL and IL.
Following this idea, we could position the calculus of possible correctness \citep{DBLP:journals/jacm/Hoare78} between NC and SIL; OL along the diagonal between HL and SIL; and the forward-oriented core set of rules in~\citet[Fig.~2]{unter2024} between IL and SIL.

Using abstract interpretation techniques, \citet{cousotpopl2024} defines a general procedure to derive proof systems starting from the underlying triple validity conditions. Most of the above  logics can be recovered in Cousot's framework by combining a small set of abstractions in different ways, as summarised in \citet[Fig.~3]{cousotpopl2024}.

\citet{DBLP:journals/pacmpl/ZhangK22} propose a calculus for quantitative weakest pre and strongest postconditions, that subsume the Boolean case. 
Using these, they propose a classification reminiscent of Definition~\ref{def:validity}. They devise notions of total/partial (in)correctness that correspond to the four logics, and draw the same correspondence between HL and NC. However, they neither develop proof systems, nor compare in detail SIL and IL with other logics.


Dynamic Logic (DL) \citep{book:dynamic-logic} is an extension of modal logic which is able to describe program properties.
It is known that HL and IL can be encoded in DL using forward box and backward diamond operator, respectively \citep{DBLP:journals/pacmpl/OHearn20}. Moreover, both SIL and NC can be encoded using forward diamond (the latter via Proposition~\ref{prop:fw-inclusion-negation-bw} and the HL encoding).
Thus, all four conditions in Definition~\ref{def:validity} can be encoded in DL, which can then be used to understand their connections and possibly get new insights.

\section{Background}\label{sec:background}
\subsection{Regular Commands.}\label{sec:reg}
Like \citet{DBLP:journals/pacmpl/OHearn20}, we consider a simple regular command language with an explicit nondeterministic operator $\regplus$.
HL can be translated in this framework as in \citet{DBLP:conf/RelMiCS/MollerOH21}.
As usual, we define arithmetic expressions $\code{a} \in \AExp$ and Boolean expressions $\code{b} \in \BExp$ as follows:
\begin{align*}
	\AExp \ni \code{a} ::= \; n \mid x \mid \code{a} \diamond \code{a} \qquad
	\BExp \ni \code{b} ::= \; \false \mid \lnot \code{b} \mid \code{b} \land \code{b} \mid \code{a} \asymp \code{a}
\end{align*}

\noindent
where $n \in \setZ$ is a natural number, $x$ is a (integer) variable, $\diamond \in \{ +, -, \cdot, \dots \}$ encodes standard arithmetic operations, and $\asymp \in \{ =, \neq, \le, <, \dots \}$ standard comparison operators.

The syntax of regular commands $\regr\in\Reg$ is:
\begin{equation}
	\ACmd \ni \regc ::= \; \code{skip} \mid \code{x := a} \mid \code{b?} \qquad
	\Reg\ni \regr ::= \; \regc \mid \regr;\regr\mid \regr \regplus \regr \mid \regr^\kstar \label{eq:reg-commands-def}
\end{equation}

Regular commands provide a general template which can be instantiated differently by changing the set $\ACmd$ of atomic commands $\regc$. The set $\ACmd$ determines the kind of operations allowed in the language, and we assume it to contain deterministic assignments and boolean guards. 
The command \code{skip} does nothing, while \code{x := a} is the standard deterministic assignment. The semantics of \code{b?} is that of an ``assume" statement: if its input satisfies \code{b} it does nothing, otherwise it diverges. At times, we also use \code{nondet()} to describe either a nondeterministic assignment (\code{x := nondet()}) or boolean expression (\code{nondet()?}).
Sequential composition is written $\regr; \regr$ and $\regr \oplus \regr$ is the nondeterministic choice. The Kleene star $\regr^\kstar$ denotes a nondeterministic iteration, where $\regr$ can be executed any number of time (possibly none) before exiting. 

This formulation can accommodate for a standard imperative programming language \cite{winskel93} by defining conditionals and \code{while}-statements as below:
\begin{align*}
	\code{if (b) \{} \code{r}_1 \code{\} else \{} \code{r}_2 \code{\}}& \ \eqdef\ (\code{b?; r}_1) \regplus ((\lnot \code{b})\code{?; r}_2) \\
	\code{while (b) \{r\}} & \ \eqdef\ (\code{b?; r})^\kstar; (\lnot \code{b})\code{?}
\end{align*}

Given our instantiation of $\ACmd$, we consider a finite set of variables $\Var$ and the set of stores $\Sigma \eqdef \Var \rightarrow \setZ$ that are (total) functions $\sigma$ from $\Var$ to integers. We tacitly assume that $0$ is the default value of uninitialised variables. Given a store $\sigma \in \Sigma$, store update $\sigma[ x \mapsto v ]$ is defined as usual for $x \in \Var$ and $v \in \setZ$. We consider an inductively defined semantics $\edenot{\cdot}$ for arithmetic and boolean expressions such that $\edenot{\code{a}}\sigma \in \setZ$ and $\edenot{\code{b}}\sigma \in \mathbb{B}$ for all $\code{a} \in \AExp$, $\code{b} \in \BExp$ and $\sigma \in \Sigma$. The semantics of atomic commands $\regc \in \ACmd$ and $S \in \wp(\Sigma)$ is defined below:
\[
\edenot{\code{skip}} S \eqdef S \qquad\quad
\edenot{\code{x := a}} S \eqdef \left\lbrace \sigma[x \mapsto \edenot{\code{a}} \sigma] \svert \sigma \in S \right\rbrace \qquad\quad
\edenot{\code{b?}} S \eqdef \left\lbrace \sigma \in S \svert \edenot{\code{b}} \sigma = \textbf{tt} \right\rbrace
\]

The collecting (forward) semantics of regular commands $\fwsem{\cdot} : \Reg \rightarrow \wp(\Sigma) \rightarrow \wp(\Sigma)$ is defined by structural induction as follows:
\begin{equation}
	\fwsem{\regc}S \eqdef \edenot{\regc} S
	\qquad
	\fwsem{\regr_1 ; \regr_2 } S \eqdef \fwsem{\regr_2}\fwsem{\regr_1} S
	\qquad
	\fwsem{\regr_1 \regplus \regr_2} S \eqdef \fwsem{\regr_1} S \cup \fwsem{\regr_2}S
	\qquad
	\fwsem{\regr^\kstar} S \eqdef \bigcup_{n \ge 0} \fwsem{\regr}^n S
	\label{eq:fwsem-definition}
\end{equation}

Roughly, $\fwsem{\regr} S$ is the set of output states reachable from the set of input states $S$. We remark that in general $\regr$ is nondeterministic and possibly non-terminating, so even when $S$ is a singleton, $\fwsem{\regr} S$ may contain none, one or many states. To shorten the notation, we write $\fwsem{\regr}\sigma$ instead of $\fwsem{\regr}\{\sigma\}$.

We define the backward semantics as the opposite relation of the forward semantics, that is
\begin{equation}
	\bwsem{\regr} \sigma' \eqdef \{ \sigma \svert \sigma' \in \fwsem{\regr} \sigma \} \label{eq:bwsem-definition}
\end{equation}
or, equivalently,
\begin{equation}
	\sigma \in \bwsem{\regr} \sigma' \iff \sigma' \in \fwsem{\regr} \sigma.  \label{eq:bwsem-sigma-sigma'}
\end{equation}
We lift the definition of backward semantics to set of states by union as usual.

The backward semantics can also be characterized compositionally, similarly to the forward one:

\begin{lemma}\label{lmm:bwsem-calculus}
	For any commands $\regr,\regr_1,\regr_2\in \Reg$ all of the following equalities hold:
	\[
	\bwsem{\regr_1; \regr_2} = \bwsem{\regr_1} \circ \bwsem{\regr_2} \qquad\qquad
	\bwsem{\regr_1 \regplus \regr_2} = \bwsem{\regr_1} \cup \bwsem{\regr_2} \qquad\qquad
	\bwsem{{\regr^\kstar}} = \bigcup\limits_{n \ge 0} \bwsem{\regr}^n
	\]
\end{lemma}

\subsection{Assertion Language}
\label{rem:fo-logic-incompleteness}
The properties of triple-based program logics depend on the expressiveness of the formulae (the assertion language) allowed in pre and postconditions~\cite{DBLP:books/ws/phaunRS01/BlassG01}.
As an example consider HL.
When the assertion language is first-order logic, HL is correct but not complete, because first-order logic is not able to represent all the properties needed to prove completeness, notably loop invariants~\citep[§2.7]{Apt81}.
If the assertion language is not close under conjunctions and disjunctions HL may be incorrect~\citep[Ex.5]{DBLP:conf/oopsla/CousotCLB12}.

To overcome the aforementioned problems, following \citet{DBLP:conf/oopsla/CousotCLB12}, we assume $P$ and $Q$ to be \emph{set} of states instead of logic formulas --- we therefore write, e.g., $\sigma \in P$ to say that the state $\sigma$ satisfies the precondition $P$.

\subsection{Hoare Logic}

\begin{figure}[t]
	\centering
	\begin{framed}
			\(
			\begin{array}{cc}
				\infer[\horule{atom}]
				{\overtriple{P}{\regc}{\fwsem{\regc}P}}
				{}
				\quad &
				\infer[\horule{cons}]
				{\overtriple{P}{\regr}{Q}}
				{P \subseteq P' & \overtriple{P'}{\regr}{Q'}& Q' \subseteq Q}
				\\[7.5pt]
				\infer[\horule{seq}]
				{\overtriple{P}{\regr_1;\regr_2}{Q}}
				{\overtriple{P}{\regr_1}{R} &
					\overtriple{R}{\regr_2}{Q}}
				\qquad &
				\infer[\horule{choice}]
				{\overtriple{P}{\regr_1 \regplus \regr_2}{Q}}
				{\forall i \in \{ 1, 2 \} & \overtriple{P}{\regr_i}{Q}}
				\\[7.5pt]
				\infer[\horule{iter}]
				{\overtriple{P}{\regr^\kstar}{P}}
				{\overtriple{P}{\regr}{P}}
			\end{array}
			\)
	\end{framed}
	\Description{Minimal set of rules for Hoare logic.}
	\caption{Hoare logic for regular commands \cite{DBLP:conf/RelMiCS/MollerOH21}}\label{fig:hl}
\end{figure}

Hoare logic (HL) \cite{hoare69} is a well known triple-based program logic which is able to prove properties about programs. The HL triple 
\[
\overtriple{P}{\regr}{Q}
\]

\noindent 
means that, whenever the execution of $\regr$ begins in a state $\sigma$ satisfying $P$ and it ends in a state $\sigma'$, then $\sigma'$ satisfies $Q$. 
When $Q$ is a correctness specification, then any HL triple $\overtriple{P}{\regr}{Q}$ provides a \emph{sufficient} condition $P$ for the so called partial correctness of the program $\regr$. Formally, this is described by the over-approximation property of postconditions
\[
\fwsem{\regr} P \subseteq Q \tag{HL},
\]

\noindent
so we say that a triple satisfying this inclusion is \emph{valid} in HL. In its original formulation, HL was given for a deterministic while-language and assuming $P$ and $Q$ were formulas in a given logic. Subsequent work generalized it to many other settings, such as nondeterminism~\citep{DBLP:journals/tcs/Apt84} and regular commands \cite{DBLP:conf/RelMiCS/MollerOH21}, resulting in the rules in Figure~\ref{fig:hl}. 
This is just a minimal core of rules which is both correct and complete, and there are many other valid rules.

\begin{theorem}[HL is sound and complete \citep{DBLP:journals/siamcomp/Cook78}]
	A HL triple is provable if and only if it is valid.
\end{theorem}

\begin{remark}[On the meaning of ``forward'']
Although the validity condition (HL) is expressed in terms of the forward semantics and thus classifies HL as a forward over-approximation (according to our terminology), the reader should not be misled to think the adjective \emph{forward} refers to direction of the deduction process, which is of course totally independent from our classification.
For example, it is well known that HL is related to Dijkstra's weakest liberal precondition~\cite{DBLP:journals/cacm/Dijkstra75} for a given postcondition $Q$: a triple $\overtriple{P}{\regr}{Q}$ is valid if and only if $P \subseteq \wlp[\regr](Q)$.
In this case, given the correctness specification $Q$, HL triples can be used for backward program analysis to find the weakest precondition for program correctness.
\end{remark}

\subsection{Incorrectness Logic}

\begin{figure}[t]
	\centering
	\begin{framed}
		\(
		\begin{array}{cc}
			\infer[\mbox{[}\mathsf{atom} \mbox{]}]
			{\undertriple{P}{\regc}{\fwsem{\regc}P}}
			{}
			\quad &
			\infer[\mbox{[}\mathsf{cons} \mbox{]}]
			{\undertriple{P}{\regr}{Q}}
			{P \supseteq P' & \undertriple{P'}{\regr}{Q'}& Q' \supseteq Q}
			\\[7.5pt]
			\infer[\mbox{[}\mathsf{seq} \mbox{]}]
			{\undertriple{P}{\regr_1;\regr_2}{Q}}
			{\undertriple{P}{\regr_1}{R} &
				\undertriple{R}{\regr_2}{Q}}
			\qquad &
			\infer[\mbox{[}\mathsf{choice} \mbox{]}]
			{\undertriple{P}{\regr_1 \regplus \regr_2}{Q_1 \cup Q_2}}
			{\forall i \in \{ 1, 2 \} & \undertriple{P}{\regr_i}{Q_i}}
			\\[7.5pt]
			\infer[\mbox{[}\mathsf{iter} \mbox{]}]
			{\undertriple{P_0}{\regr^\kstar}{\bigcup\limits_{n \ge 0} P_n}}
			{\forall n \ge 0 \sdot \undertriple{P_n}{\regr}{P_{n+1}}}
			&
		\end{array}
		\)
	\end{framed}
	\Description{Minimal set of rules for Incorrectness Logic.}
	\caption{Minimal Incorrectness Logic for regular commands \cite{DBLP:conf/RelMiCS/MollerOH21}}\label{fig:reverse-hl}
\end{figure}

Dually to HL, Incorrectness Logic \cite{DBLP:journals/pacmpl/OHearn20} was introduced as a formalism for under\hyp{}approximation with the idea of finding true bugs in the code. The IL triple 
\[\undertriple{P}{\regr}{Q}\]

\noindent
means that all the states in $Q$ are reachable from states in $P$. This is described by the property 
\[
\fwsem{\regr} P \supseteq Q, \tag{IL}
\]

\noindent
which characterises \emph{valid} IL triples. In other words, any valid IL triple $\undertriple{P}{\regr}{Q}$ states that for any state $\sigma' \in Q$ there is a state $\sigma \in P$ such that $\sigma'$ is reachable from $\sigma$. This means that if $\sigma'\in Q$ is an ``error" state, the triple guarantees it is a true alarm of the analysis. The rules, inspired by previous work on reverse Hoare logic~\cite{de-vries-koutavas-reverse-hoare-logic} and here simplified to not separate correct and erroneous termination states, are shown in Figure~\ref{fig:reverse-hl}. Just like the HL rules in Figure~\ref{fig:hl}, this is a minimal set of correct and complete rules for IL.

\begin{theorem}[IL is sound and complete \citep{DBLP:journals/pacmpl/OHearn20}]
	An IL triple is provable if and only if it is valid.
\end{theorem}

\begin{example}
	HL and IL aim at addressing different properties. Even the simple, deterministic, terminating program $\regr 42$ of Figure~\ref{fig:simple_example} can be used to show their difference.
	Recall that  $Q_{42} \eqdef (z = 42)$ denotes the set of erroneous states, i.e., $Q_{42}$ is an incorrectness specification.
	The valid HL triple 
	\[
	\overtriple{ \odd(y) \wedge \even(x)}{\regr 42}{Q_{42}}
	\]
	
	\noindent
	identifies input states that will surely end up in an  error state, while the triple
	\[
	\overtriple{ \neg Q_{42} \wedge(\odd(x) \vee \even(y))}{\regr 42}{\neg Q_{42}}
	\]
	
	\noindent
	 characterize input states that will not produce any error. 
	On the other hand, the valid IL triple 
	\[
	\undertriple{\true}{\regr 42}{Q_{42} \wedge \odd(y) \wedge \even(x) }
	\]
	
	\noindent
	expresses the fact that error states in $Q_{42}$ are reachable by some initial state. Similarly, also the IL triple 
	\[
	\undertriple{\true}{\regr 42}{\neg Q_{42} \wedge(\odd(x)\vee \even(y)) }
	\]
	
	\noindent
	is valid since the postcondition $\neg Q_{42}$ can be reached only when the path conditions to reach the assignment are not satisfied.
	\qed
\end{example}

\subsection{Necessary Conditions}
For contract inference, \citet{DBLP:conf/vmcai/CousotCFL13,DBLP:conf/vmcai/CousotCL11} introduced the notion of Necessary Conditions (NC).
The goal was to relax the burden on programmers: while sufficient conditions require the caller of a function to supply parameters that will never cause an error, NC only prevents the invocation of the function with arguments that will inevitably lead to some error.
Intuitively, for $Q$ the set of good final states, the NC triple 
\[\necctriple{P}{\regr}{Q}\]

\noindent
means that any state $\sigma \in P$ admits at least one non-erroneous execution of the program $\regr$.
Recently, the same concept has been applied to the context of security \citep{DBLP:journals/pacmpl/MackayEND22}.

Following the original formulation~\cite{DBLP:conf/vmcai/CousotCFL13}, we can partition the traces of a nondeterministic execution starting from a memory $\sigma$ in three different sets: $\mathcal{T}(\sigma)$, those without errors, $\mathcal{E}(\sigma)$, those with an error, and $\mathcal{I}(\sigma)$, those which do not terminate.
A sufficient precondition $\overline{P}$ is such that $(\sigma\models \overline{P}) \implies (\mathcal{E}(\sigma) = \emptyset)$, that is $\overline{P}$ excludes all error traces.
Instead, a necessary precondition $\underline{P}$ is a formula such that $(\mathcal{T}(\sigma) \neq \emptyset \lor \mathcal{I}(\sigma) \neq \emptyset) \implies (\sigma \models \underline{P})$, which is equivalent to
\[
(\sigma \not\models \underline{P}) \implies (\mathcal{T}(\sigma) \cup \mathcal{I}(\sigma) = \emptyset) .
\]

\noindent In other words, a necessary precondition \emph{rules out no good run}: when it is violated by the input state, the program has only erroneous executions. 
Note that we consider infinite traces as good traces. We do this by analogy with sufficient preconditions, where bad traces are only those which end in a bad state.

\begin{example}\label{ex:nc-running}
	Consider again the nondeterministic program $\mathsf{r42nd}$ of Figure~\ref{fig:simple_example_nondet}, and the correctness specification $\neg Q_{42} = (z \neq 42)$. Then
	\begin{center}
		\begin{tabular}{c|cc}
			& $\mathcal{T}(\sigma)$ & $\mathcal{E}(\sigma)$ \\
			\hline
			$\sigma(z)\neq 42$ & $\neq \emptyset$ & $\neq \emptyset$ \\
			$\sigma(z)=42$ & $\emptyset$ & $\neq \emptyset$ \\
		\end{tabular}
	\end{center}
	The weakest sufficient precondition for this program is $\overline{P} = (z\neq 42 \wedge \even(y))$ because no input state $\sigma$ that violates $\overline{P}$ is such that $\mathcal{E}(\sigma) = \emptyset$.
	On the contrary, we have, e.g., that $\underline{P}= (z\neq 42)$ is a necessary precondition, while $(z > 42)$ is not, because it excludes some good runs.
	\qed
\end{example}

\section{Sufficient Incorrectness Logic}\label{sec:sil}
In this section, we explain the rationale for SIL, give a minimal sound and complete set of rules and present some additional rules that can speed up program analysis.

\subsection{Rationale}

As outlined in the Introduction, the main motivation for studying sufficient incorrectness conditions is to enhance program analysis frameworks with the ability to identify the source of incorrectness and provide programmers with evidence of the inputs that can cause the errors. This is different from IL, which can be used to determine the presence of bugs, but cannot precisely identify their sources. SIL is also different in spirit from HL and NC, that aim to prove the absence of bugs or prevent them.
Roughly, the SIL triple 
\[\angletriple{P}{\regr}{Q}\]

\noindent
requires that for all $\sigma \in P$, there exists at least one $\sigma' \in Q$ such that $\sigma \in \bwsem{\regr} \sigma'$ or, equivalently, $\sigma' \in \fwsem{\regr} \sigma$.
In other words, any SIL triple $\angletriple{P}{\regr}{Q}$ asserts that ``\emph{all states in $P$ have at least one execution leading to a state in $Q$}". 
More concisely, this property amounts to the validity condition
\[
\bwsem{\regr} Q \supseteq P. \tag{SIL}
\]
Please note that, in the presence of nondeterminism, states in $P$ are required to have one execution leading to $Q$, not necessarily all of them.
We make this equivalence formal by means of the following proposition.

\begin{proposition}[Characterization of SIL validity]\label{prop:sil-validity-characterization}
	For any $\regr \in \Reg$, $P, Q \subseteq \Sigma$ we have
	\[
	\bwsem{\regr} Q \supseteq P \iff \forall \sigma \in P \sdot \exists \sigma' \in Q \sdot \sigma' \in \fwsem{\regr} \sigma
	\]
\end{proposition}

A convenient way to exploit SIL is to assume that the analysis takes as input the incorrectness specification $Q$, which represents the set of erroneous final states.
Then, any valid SIL triple $\angletriple{P}{\regr}{Q}$ yields a precondition which surely captures erroneous executions. In this sense, $P$ can be considered a sufficient condition for incorrectness. This is dual to the interpretation of IL, where, for a given precondition $P$, any IL triple $\undertriple{P}{\regr}{Q}$ yields a set $Q$ of final states which are for sure reachable, so that any error state in $Q$ is a true bug reachable from some input in $P$.

\subsection{Proof System}\label{sec:SILrules}
\begin{figure}[t]
	\centering
	\begin{framed}
		\(
		\begin{array}{cc}
			\infer[\lrule{atom}]
			{\angletriple{\bwsem{\regc}Q}{\regc}{Q}}
			{}
			\quad &
			\infer[\lrule{cons}]
			{\angletriple{P}{\regr}{Q}}
			{P \subseteq P' & \angletriple{P'}{\regr}{Q'}& Q' \subseteq Q}
			\\[7.5pt]
			\infer[\lrule{seq}]
			{\angletriple{P}{\regr_1;\regr_2}{Q}}
			{\angletriple{P}{\regr_1}{R} &
				\angletriple{R}{\regr_2}{Q}}
			\qquad &
			\infer[\lrule{choice}]
			{\angletriple{P_1 \cup P_2}{\regr_1 \regplus \regr_2}{Q}}
			{\forall i \in \{ 1, 2 \} & \angletriple{P_i}{\regr_i}{Q}}
			\\[7.5pt]
			\infer[\lrule{iter}]
			{\angletriple{\bigcup\limits_{n \ge 0} Q_n}{\regr^\kstar}{Q_0}}
			{\forall n \ge 0 \sdot \angletriple{Q_{n+1}}{\regr}{Q_n}}
		\end{array}
		\)
	\end{framed}
	\Description{Minimal set of rules for Sufficient Incorrectness Logic.}
	\caption{Sufficient Incorrectness Logic}\label{fig:sil}
\end{figure}

We present the proof rules for SIL in Figure~\ref{fig:sil}. This set of rules is minimal, correct and complete (Section~\ref{sec:sil-sound-and-complete}); additional valid rules are discussed in Section~\ref{sec:SILmorerules}.

For the \code{skip} atomic command, all the states in $Q$ will reach $Q$ itself.
For assignments, the sufficient precondition is given by the backward semantics applied to $Q$ (cf.~(\ref{eq:bwsem-definition})).
For Boolean guards, the set of initial states which can lead to $Q$ are all those states in $Q$ that also satisfies the guard.
These three cases are summarized in the rule \lrule{atom}.

If we know that all states in $R$ have an execution of $\regr_2$ ending in $Q$ and all states in $P$ have an execution of $\regr_1$ ending in $R$, we can deduce that all states in $P$ have an execution of $\regr_1; \regr_2$ ending in $Q$. This is captured by rule \lrule{seq}.

If all states in $P_1$ have an execution of $\regr_1$ ending in $Q$, they also have an execution of $\regr_1 \regplus \regr_2$ ending in $Q$ since its semantics is a superset of that of $\regr_1$ (cf.~(\ref{eq:fwsem-definition})). The reasoning for $P_2$ is analogous, thus yielding the rule \lrule{choice}. This rule is reminiscent of the equation for conditionals in the calculus of possible correctness \citep{DBLP:journals/jacm/Hoare78}.

For iterations, if the command $\regr$ is never executed the precondition is trivially the postcondition $Q_0$. If the command is executed $n > 0$ times, then the precondition is the union of the preconditions of every iteration step.
This is formalized by rule \lrule{iter}.
We remark that this rule first appeared in \citet[§5]{DBLP:conf/RelMiCS/MollerOH21}.

If all the states in $P'$ have an execution leading to a state in $Q'$, then any subset $P$ of $P'$ will lead to $Q'$ as well.
Similarly, every superset $Q$ of $Q'$ is reachable by executions starting in $P$.
Those two observations lead to \lrule{cons}, which is the key rule of SIL as it enables the under-approximation of the precondition.

\subsection{Correctness and Completeness}\label{sec:sil-sound-and-complete}
SIL is both correct and complete.
Correctness can be proved by induction on the derivation tree of a triple. Intuitively, if the premises of a rule are valid, then its consequence is valid as well, as we briefly observed in the previous section.
To prove completeness, we rely on the fact that rules other than \lrule{cons} are exact, that is, if their premises satisfy the equality $\bwsem{\regr} Q = P$, their conclusion does as well. Using this, we prove the triple $\angletriple{\bwsem{\regr}Q}{\regr}{Q}$ for any $\regr$ and $Q$. Then we conclude using \lrule{cons} to get a proof of $\angletriple{P}{\regr}{Q}$ for any $P \subseteq \bwsem{\regr}Q$.
This is formalized by:

\begin{theorem}[SIL is sound and complete]\label{thm:sil-sound-complete}
	A SIL triple is provable if and only if it is valid.
\end{theorem}

\subsection{Additional Rules for Program Analysis}\label{sec:SILmorerules}
\begin{figure}[t]
	\centering
	\begin{framed}
		\(
		\begin{array}{cc}
			\infer[\lrule{empty}]
			{\angletriple{\emptyset}{\regr}{Q}}
			{}
			\qquad &
			\infer[\lrule{disj}]
			{\angletriple{P_1 \cup P_2}{\regr}{Q_1 \cup Q_2}}
			{\angletriple{P_1}{\regr}{Q_1} & \angletriple{P_2}{\regr}{Q_2} }
			\\[7.5pt]
			\infer[\lrule{iter0}]
			{\angletriple{Q}{\regr^\kstar}{Q}}
			{}
			\qquad &
			\infer[\lrule{unroll}]
			{\angletriple{P}{\regr^\kstar}{Q}}
			{\angletriple{P}{\regr^\kstar; \regr}{Q}}
			\\[7.5pt]
			\infer[\lrule{unroll\mbox{-}split}]
			{\angletriple{P \cup Q_2}{\regr^\kstar}{Q_1 \cup Q_2}}
			{\angletriple{P}{\regr^\kstar; \regr}{Q_1}}
		\end{array}
		\)
	\end{framed}
	\Description{Additional valid rules for SIL. While not required for completeness, these rules can help shortening proofs.}
	\caption{Additional rules for SIL}\label{fig:sil-extra}
\end{figure}

SIL proof system in Figure~\ref{fig:sil} is deliberately minimal: if we remove any rule it is no longer complete. However, there are many other valid rules which can be useful for program analysis and cannot be derived from the five we presented. Some of them are in Figure~\ref{fig:sil-extra}.

Rule \lrule{empty} is used to drop paths backward, just like IL can drop them forward (and analogous axiom $\undertriple{P}{\regr}{\emptyset}$ is valid for IL). Particularly, this allows to ignore one of the two branches with \lrule{choice}, or to stop the backward iteration of \lrule{iter} without covering all the infinite iterations.
An example of such an application is the derived rule \lrule{iter0}: it can be proved from rules \lrule{iter} and \lrule{empty} by taking $Q_0 = Q$ and $Q_n = \emptyset$ for $n \ge 1$. This rule corresponds to not entering the iteration at all. It subsumes HL's rule \horule{iter}, which is based on loop invariants: those are a correct but not complete reasoning tool for under-approximation \citep{DBLP:journals/pacmpl/OHearn20}.

\begin{figure}
\begin{subfigure}[t]{0.43\textwidth}
\begin{lstlisting}[language=C]
// program rshortloop0

n := nondet();
while(n > 0) {
    x := x + n;
    n := nondet();
}
// assert(x != 2000000)
\end{lstlisting}
\caption{An example to illustrate the use of \lrule{iter0}.}
\label{fig:rshortloop0}
\end{subfigure}
\hfill
\begin{subfigure}[t]{0.44\textwidth}
\begin{lstlisting}[language=C]
// program rloop0
x := 0;
n := nondet();
while(n > 0) {
    x := x + n;
    n := nondet();
}
// assert(x != 2000000)
\end{lstlisting}
\caption{An example to illustrate the use of \lrule{unroll}.}\label{fig:rloop0}
\end{subfigure}
\caption{Two sample programs adapted from \citet[\S 6.1]{DBLP:journals/pacmpl/OHearn20}.}
\end{figure}

\begin{example}\label{ex:sil-derivation-1}
	To show the use of \lrule{iter0}, we consider the program $\mathsf{rshortloop0}$ in Figure~\ref{fig:rshortloop0}.
	It is a slight variation of the program presented as ``loop0" in \citet[\S 6.1]{DBLP:journals/pacmpl/OHearn20}, where the final error states were $Q_{2M} \eqdef (x = 2 000 000)$. We can write this program in the syntax of regular commands as
	\[
	\mathsf{rshortloop0} \eqdef \code{n := nondet(); }(\code{(n > 0)?; x := x + n; n := nondet()})^{\kstar}; \code{(n <= 0)?}
	\]
	
	To prove a SIL triple for $\mathsf{rshortloop0}$, we let 
	\[
	\regr_w \eqdef \code{(n > 0)?; x := x + n; n := nondet()}
	\]
	
	\noindent
	be the body of the loop and $R_{2M} \eqdef (x = 2000000 \land n \le 0) = (Q_{2M} \land n \le 0)$. We perform the following derivation:
	
	{
	\footnotesize
	\[
	\infer[\lrule{seq}]
	{\angletriple{Q_{2M}}{\mathsf{rshortloop0}}{Q_{2M}}}
	{
		\infer[\lrule{atom}]{\angletriple{Q_{2M}}{\code{n := nondet()}}{R_{2M}}}{}
		&
		\infer[\lrule{seq}]{\angletriple{R_{2M}}{(\regr_w)^{\kstar}; \code{(n <= 0)?}}{Q_{2M}}}{
			\infer[\lrule{iter0}]{\angletriple{R_{2M}}{(\regr_w)^{\kstar}}{R_{2M}}}{}
			&
			\infer[\lrule{atom}]{\angletriple{R_{2M}}{\code{(n <= 0)?}}{Q_{2M}}}{}
		}
	}
	\]
	}
	
	It is worth noticing that we use the rule \lrule{iter0} to bypass the iteration: this is possible because we do not need to enter the loop to have a path reaching the entry point of the program.
	\qed
\end{example}

Rule \lrule{unroll} allows to unroll a loop once. Subsequent applications of this rule allow to simulate (backward) a finite number of iterations, and then rule \lrule{iter0} can be used to ignore the remaining ones. This is on par with IL ability to unroll a loop a finite number of times to find some postcondition. Please note that rules \lrule{iter0} and \lrule{unroll} are analogous to rules \irule{iter0} and \irule{unroll} (there called \emph{Iterate non-zero}) of IL \citep{DBLP:journals/pacmpl/OHearn20,DBLP:conf/RelMiCS/MollerOH21}.

Rule \lrule{disj} allows to split the analysis and join the results, just like HL and IL.
However, while a corresponding rule \lrule{conj} which perform intersection is sound for HL, it is unsound for both IL and SIL (see Example~\ref{ex:il-no-strongest-pre}).

The last rule \lrule{unroll\mbox{-}split} is derived from \lrule{unroll} and \lrule{disj}. The intuition behind this rule is to split the postcondition $Q$ in two parts, $Q_1$ which goes through the unrolled loop and $Q_2$ which instead skips the loop entirely; the results are then joined. This rule is analogous to rule \olrule{Induction} of OL \citep{DBLP:journals/corr/ZilbersteinDS23}.

\begin{example}\label{ex:sil-derivation-2}
	To illustrate the use of \lrule{unroll}, let us
	consider the full version of ``loop0" from \citet[\S 6.1]{DBLP:journals/pacmpl/OHearn20}, which is reported in Figure~\ref{fig:rloop0}.
	We can translate it in the language of regular commands by letting
	\[
	\mathsf{rloop0} \eqdef \code{x := 0;} \mathsf{rshortloop0}
	\]
	where $\mathsf{rshortloop0}$ is the one defined in Example~\ref{ex:sil-derivation-1}. Final error states are those in $Q_{2M}$.
	To prove a triple for $\mathsf{rloop0}$, we can no longer ignore the loop with \lrule{iter0} because the initial assignment \code{x := 0} would yield the precondition $\false$ on $Q_{2M}$ if we wanted to extend the proof of $\angletriple{Q_{2M}}{\mathsf{rshortloop0}}{Q_{2M}}$ from the previous example. In this case, we need to perform at least one iteration, and we do so using \lrule{unroll}.
	
\begin{figure}[t]
	\centering
	\footnotesize
	\[
		\infer[]
		{(*)}
		{
			\infer[\lrule{seq}]{\angletriple{T_{2M}}{\regr_w^{\kstar}; \code{(n <= 0)?}}{Q_{2M}}}{
				\infer[\lrule{unroll}]{\angletriple{T_{2M}}{\regr_w^{\kstar}}{R_{2M}}}{
					\infer[\lrule{seq}]{\angletriple{T_{2M}}{\regr_w^{\kstar}; \regr_w}{R_{2M}}}{
						\infer[\lrule{iter0}]{\angletriple{T_{2M}}{\regr_w^{\kstar}}{T_{2M}}}{}
						&
						\infer{\angletriple{T_{2M}}{\regr_w}{R_{2M}}}{\vdots}
					}
				}
				&
				\infer[\lrule{atom}]{\angletriple{R_{2M}}{\code{(n <= 0)?}}{Q_{2M}}}{}
			}
		}
	\]

	\[
		\infer[\lrule{seq}]
		{\angletriple{\true}{\mathsf{rloop0}}{Q_{2M}}}
		{
			\infer[\lrule{seq}]{\angletriple{\true}{\code{x := 0; n := nondet()}}{T_{2M}}}{
				\infer[\lrule{atom}]{\angletriple{\true}{\code{x := 0}}{x \le 2000000}}{}
				&
				\infer[\lrule{atom}]{\angletriple{x \le 2000000}{\code{n := nondet()}}{T_{2M}}}{}
			}
			&
			(*)
		}
	\]
	\caption{Derivation of the SIL triple $\angletriple{\true}{\mathsf{rloop0}}{Q_{2M}}$ for Example~\ref{ex:sil-derivation-2}.}
	\label{fig:sil-example-derivation-2}
\end{figure}
	
	We let $\regr_w$ and $R_{2M}$ be as in Example~\ref{ex:sil-derivation-1}, and we define $T_{2M} \eqdef (x + n = 2000000 \land n > 0)$. We show the derivation in Figure~\ref{fig:sil-example-derivation-2}, where we omit the proof of $\angletriple{T_{2M}}{\regr_w}{R_{2M}}$ since it is a straightforward application of \lrule{seq} and \lrule{atom}.
	Please note that the pattern we used to prove $\angletriple{T_{2M}}{\regr_w^{\kstar}}{R_{2M}}$ given the proof of $\angletriple{T_{2M}}{\regr_w}{R_{2M}}$ is general and does not depend on our specific proof. It corresponds to a single loop unrolling: if we know that one execution of the loop body satisfies some given pre and postconditions, the same holds for its Kleene iteration, because of the under-approximation. This is not a new result: it was already observed for IL.
	\qed
\end{example}

It is worth noticing that a version of the logic where \lrule{unroll} and \lrule{disj} replace \lrule{iter} is not complete. The reason is that \lrule{unroll} is not able to perform an infinite amount of iterations with a finite derivation tree, thus preventing it from proving executions which require an infinite amount of iterations to reach the fixpoint.

\subsection{SIL and Outcome Logic}\label{sec:silvsol}
As discussed in the Introduction, OL already recognized the importance of locating the source of errors, and showed that any SIL triple $\angletriple{P}{\regr}{Q}$ is valid if and only if the corresponding \NDOL{} (the nondeterministic instance of OL, see Section~\ref{sec:sil-and-ol}) triple $\oltriple{P}{\regr}{Q \oplus \top}$ is valid.
However, the two approaches differ in the proof systems: in fact, many SIL rules cannot be derived from \NDOL{} rules.

The first thing to observe is that \NDOL{} assertions use a higher-level language, which describes sets of subsets of $\Sigma$, while SIL assertions are just subsets of $\Sigma$. In the \NDOL{} assertion language, subsets of $\Sigma$ (i.e., SIL assertions) are called \emph{atomic assertions}, which are then combined with higher level operators to denote sets of sets of states. 
To see this, take the two atomic assertions $P_x \eqdef (x = 0)$ and $P_y \eqdef (y = 0)$.  In SIL, their conjunction $P_x \cap P_y$ is satisfied by any state $\sigma$ such that $\sigma(x)=\sigma(y)=0$, while in \NDOL{} the compound expression $P_x \wedge P_y$ is satisfied by any set of states $m$ such that  $\forall\sigma\in m. \sigma(x)=\sigma(y)=0$. 
This means that the correspondence between SIL and \NDOL{} triples can only be drawn when the latter is in the specific form $\oltriple{P}{\regr}{Q \oplus \top}$ with $P$, $Q$ atomic assertions, while in general a proof in OL involves assertions not in that form. Thus, there is no simple way to translate a proof in \NDOL{} to a proof in SIL.
As discussed in Section~\ref{sec:sil-and-ol}, disjunction further distinguishes the assertion languages: the union $P_x \cup P_y$ is the only kind of disjunction available in SIL, while \NDOL{} has also $\lor$ and $\oplus$.

Some SIL rules can be derived from OL ones. For instance, rule \lrule{seq} can be derived from \olrule{Seq}, and \lrule{atom} is a combination of \olrule{One}, \olrule{Assign} and \olrule{Assume}.
However, we could not derive three key rules of SIL, namely \lrule{choice}, \lrule{disj} and \lrule{iter}.
All these rules depend on the union of atomic assertions, which cannot be mentioned in OL because it may not exist in the outcome monoid. 

Extending the syntax of outcome predicates and exploiting some rules that are much more similar to ours, a preprint~\citep{zilberstein2024relatively} presents a revised version of OL and recovers completeness for Lisbon triples. 
However, since it mixes forward and backward-oriented rules, mixes over- and under-approximations and addresses more general (weighted) settings, 
we claim that the direct use of SIL rules is more straightforward and convenient.
This advantage is more evident when we instantiate SIL for Separation Logic: in the small axioms for atomic commands, SIL can exploit an assertion language that is expressive enough to guarantee relative completeness, whereas the corresponding instance of OL \citep{zilberstein2023outcome} is not complete.

\subsection{Industrial Applications}\label{sec:tool}
SIL represents \emph{a posteriori} formalization and theoretical justification of the parallel, modular, and compositional static analysis implemented in industrial grade static analyzers for security developed and used at Meta, such as Zoncolan~\cite{DBLP:journals/cacm/DistefanoFLO19}, Mariana Trench~\cite{MarianaTrench}, and Pysa~\cite{Pysa}.
Those tools automatically find more than $50$\% of the security bugs in the Meta family of apps and many SEVs ~\citep[Fig.~5]{DBLP:journals/cacm/DistefanoFLO19}.

In order to scale up to hundreds of millions of lines of code, static analyses need to be parallelizable and henceforth modular and compositional.
Modularity implies that the analysis can infer \emph{meaningful} information without full knowledge of the global program.
Compositionality means that the~\emph{results} of analyzing modules are good enough that one does not lose information in using the inferred triple instead of inlining the code. 

The analysis implemented in the aformentioned tools is a modular backward analysis that determines which input states for a function will lead to a security error, likewise the SIL rules described in this paper.
In particular, the analysis infers sufficient incorrectness preconditions (modularity) for callees that can be used by the callers (compositionality) to generate their incorrectness preconditions. 
When the propagation of the inferred sufficient precondition reaches an attacker\hyp{}controlled input, the analyzers check if that input is included in the propagated error condition.
If it is the case, then it emits an error. 
The function analyses are parallelized and a strategy similar to the iteration rule of Figure~\ref{fig:sil} is used to compute the fixpoint in presence of mutually recursive functions.

\section{Relations Among Logics}\label{sec:comparison}
We follow the two-dimensional scheme in Figure~\ref{fig:square-full} to carry out an exhaustive comparison among the four validity conditions. While the duality between HL and IL was crystal clear from the introduction of IL in~\citet{DBLP:journals/pacmpl/OHearn20}, to the best of our knowledge its relations with NC has not been explored yet. Pursuing such formal comparison leads to some surprising results:
\begin{itemize}
	\item Although NC and IL share some similarities, not only we show that they are not comparable, but we prove a bijective correspondence between NC and HL that exploits double negation. 
	\item While the classification in Figure~\ref{fig:square-full} and the previous item suggests that the proof strategy used to define the bijective correspondence between NC and HL can be extended to relate IL and SIL, we show that this is not the case: IL and SIL are not comparable.
	\item More in general, each one of the conditions (and the corresponding logics) focus on different aspects of programs: none of them is better than the other; instead, each one has its precise application.
\end{itemize}

Next, we discuss the classification scheme, we deepen the comparison between NC and other logics, we show how SIL and IL differ, and we extend the discussion along other axes of comparison.

\subsection{Taxonomy}

The validity of HL and IL triples correspond to the constraints $\fwsem{\regr} P \subseteq Q$ and $\fwsem{\regr} P \supseteq Q$, respectively.
Therefore, their validity is expressed using the forward semantics. The difference between the two inequalities resides in the direction of the inclusion: the postcondition is an over-approximation of the forward semantics in HL, while it is an under-approximation in IL.
Thus we can classify HL and IL along the first column of Figure~\ref{fig:square-full}.
Dually, the validity condition for SIL triples, namely $\bwsem{\regr} Q \supseteq P$, can be placed in the second column, since it is expressed in terms of backward semantics, and in the second row, because the precondition is an under-approximation of the backward semantics.

The fourth inequality $\bwsem{\regr} Q \subseteq P$ defines over-approximations of backward semantics, and is dual to both SIL and HL. It turns out that $\bwsem{\regr} Q \subseteq P$ is the validity condition for NC.
Indeed, a precondition $P$ is necessary for $Q$ if every state which can reach a state in $Q$ is in $P$. If $Q$ describes good final states, then $\bwsem{\regr} Q$ defines all states $\sigma$ which can reach a good state, so every necessary precondition $\underline{P}$ must contain at least all states in $\bwsem{\regr} Q$.
More formally, given a final state $\sigma' \in Q$, all traces starting in $\sigma$ and ending in $\sigma'$ are in $\mathcal{T}(\sigma)$. This means that every state $\sigma \in \bwsem{\regr} \sigma'$ has a trace in $\mathcal{T}(\sigma)$ (the one ending in $\sigma'$), so it must belong to a necessary precondition $\underline{P}$ as well.
This is captured by the following:

\begin{proposition}[NC as backward over-approximation]\label{prop:nc}
Given a correctness postcondition $Q$ for the program $\regr$, any possible necessary precondition $\underline{P}$ for $Q$ satisfies:
\[
\bwsem{\regr} Q \subseteq \underline{P} \tag{NC}.
\]
\end{proposition}

An analogous characterization appeared in \citet[§6.3]{DBLP:journals/pacmpl/ZhangK22} in terms of weakest precondition.
We will use (NC) to understand how NC relates to IL (Section~\ref{sec:nc-il}) and HL (Section~\ref{sec:nc-hl}).

\subsection{Necessary Preconditions and Incorrectness Logic}\label{sec:nc-il}
Sufficient preconditions are properties implying the weakest liberal precondition: $\overline{P}$ is sufficient for a postcondition $Q$ if and only if $\overline{P} \implies \wlp[\regr](Q)$, which in turn is equivalent to validity of the HL triple $\overtriple{\overline{P}}{\regr}{Q}$. Necessary and sufficient preconditions are dual, and so are IL and HL. Moreover, NC and IL enjoy the same consequence rule: both can strengthen the postcondition and weaken the precondition. This suggest that there may be a relation between NC and IL. However, the following example shows this is not the case.

\begin{example}
Let us consider once again the simple deterministic program $\regr 42$ of Figure~\ref{fig:simple_example}, with $Q_{42} \eqdef (z=42)$.
Let $Q'_{42} \eqdef (Q_{42} \wedge \odd(y) \wedge \even(x))$ and $P_{11} \eqdef (z = 11)$. Then IL triple~(\ref{eq:il-simple-det}) is exactly $\undertriple{P_{11}}{\regr 42}{Q'_{42}}$, so we know this is valid.
However, we observe that the NC triple $\necctriple{P_{11}}{\regr 42}{Q'_{42}}$ is not valid because the state $\sigma$ such that $\sigma(x) = 0$, $\sigma(y) = 1$, $\sigma(z) = 10$ has an execution leading to $Q'_{42}$ but is not in $P_{11}$.
Moreover, take for instance $\underline{P} = (\odd(y) \land \even(x))$, which makes the NC triple $\necctriple{\underline{P}}{\regr 42}{Q'_{42}}$ valid (any state \emph{not} satisfying $\underline{P}$ has either an even $y$ or an odd $x$, and those variables are not changed by the program). Then it is clear that $P_{11} \not\implies \underline{P}$. This shows that not only IL triples do not yield NC triples, but also that in general there are NC preconditions which are not implied by IL preconditions.

Conversely, consider $\neg Q_{42} = (z \neq 42)$. The NC triple $\necctriple{\true}{\regr 42}{\neg Q_{42}}$ is clearly valid, but the IL triple $\undertriple{\true}{\regr 42}{\neg Q_{42}}$ is not: for instance, the final state $\sigma'$ such that $\sigma'(x) = \sigma'(y) = \sigma'(z) = 11$ is not reachable from any initial state. It follows that the IL triple $\undertriple{P}{\regr 42}{\neg Q_{42}}$ is not valid for any $P$.
\qed
\end{example}

Given a valid NC triple $\necctriple{\underline{P}}{\regr}{Q}$ and a valid IL triple $\undertriple{P}{\regr}{Q}$, we are guaranteed that $P \cap \underline{P} \neq \emptyset$, that is there are states which satisfy both $P$ and $\underline{P}$. However, in general neither $P \subseteq \underline{P}$ nor $\underline{P} \subseteq P$ hold.
The difference between NC and IL is striking when we look at their quantified definitions:
\begin{align*}
\forall \sigma' \in Q \sdot \forall \sigma\in \bwsem{\regr}\sigma' \sdot\sigma\in P 
\tag{NC$^{\forall}$} \\
\forall \sigma' \in Q \sdot \exists \sigma\in \bwsem{\regr}\sigma' \sdot \sigma \in P 
\tag{IL$^{\exists}$}
\end{align*}

\noindent
While (NC$^{\forall}$) universally quantifies on initial states (\emph{all} initial states with a good run must satisfy the precondition), (IL$^{\exists}$) existentially quantifies on them.
We observe that, whenever $\regr$ is reversible (i.e., $\fwsem{\regr}$ is injective) any valid IL triple is also a valid NC triple.

\subsection{Necessary Preconditions and Hoare Logic}\label{sec:nc-hl}
It turns out that NC is strongly connected to weakest liberal preconditions and thus to HL.
In fact, let $Q$ be a postcondition: a finite trace is in $\mathcal{T}(\sigma)$ if its final state satisfies $Q$ and in $\mathcal{E}(\sigma)$ otherwise.
In general, a necessary precondition has no relationship with $\wlp[\regr](Q)$. However, if we consider $\lnot Q$ instead of $Q$, we observe that ``erroneous" executions becomes those in $\mathcal{T}(\sigma)$ and ``correct" ones those in $\mathcal{E}(\sigma)$. This means that
\[
(\mathcal{T}(\sigma) = \emptyset) \iff \sigma \in \wlp[\regr](\lnot Q),
\]
from which we derive 
\[\lnot \underline{P} \implies \wlp[\regr](\lnot Q)\] 
or, equivalently,
\[\lnot \wlp[\regr](\lnot Q) \implies \underline{P} .\]

\begin{example}
Building on Example~\ref{ex:nc-running}, the correctness specification for $\mathsf{r42nd}$ is $\neg Q_{42} = (z\neq 42)$. We have that $\wlp[\mathsf{r42nd}](\lnot \lnot Q_{42}) = \wlp[\mathsf{r42nd}](Q_{42}) = Q_{42}$, because if initially $z \neq 42$ then there is the possibility that $x$ is assigned an odd value and $z$ is not updated.
Therefore, a condition $P$ is implied by $\lnot \wlp[\regr](\lnot\lnot Q_{42}) = \lnot Q_{42}$ if and only if it is necessary. For instance, $(z\neq 42 \vee \odd(y))$ is necessary, while $(z > 42)$ is not.
\qed
\end{example}

The next bijection establishes the connection between NC and HL:

\begin{proposition}[Bijection between NC and HL]\label{prop:fw-inclusion-negation-bw}
	For any  $\regr\in \Reg$ and $P,Q\subseteq \Sigma$ we have:
	\[\fwsem{\regr} P \subseteq Q \iff \bwsem{\regr} (\lnot Q) \subseteq \lnot P .\]
\end{proposition}

The proposition means that a necessary precondition is just the negation of a sufficient precondition for the negated post.
This was also observed using weakest (liberal) preconditions in \citet[Theorem~5.4]{DBLP:journals/pacmpl/ZhangK22}.

\subsection{Sufficient Incorrectness Logic and Hoare Logic}
In general, HL and SIL are different logics, but they coincide whenever the program $\regr$ is deterministic and terminates for every input (cf. Section~\ref{sec:sil-and-hl}). This is formally captured by:

\begin{proposition}[SIL vs HL]\label{prop:sil-hl-deterministic-terminating}
	For any $\regr\in \Reg$ and $P, Q \subseteq \Sigma$ we have:
	\begin{itemize}
		\item if $\regr$ is deterministic, $\bwsem{\regr} Q \supseteq P \implies \fwsem{\regr} P \subseteq Q$
		\item if $\regr$ is terminating, $\fwsem{\regr} P \subseteq Q \implies \bwsem{\regr} Q \supseteq P$
	\end{itemize}
\end{proposition}

\subsection{Sufficient Incorrectness Logic and Incorrectness Logic}

In Figure~\ref{fig:square-full}, we highlight the fact that (HL) and (NC) are isomorphic (Proposition~\ref{prop:fw-inclusion-negation-bw}). It is natural to ask if there is a similar connection between (IL) and (SIL). The next example answers negatively.

\begin{example}\label{ex:il-sil-incomparable}
	Since IL and SIL enjoy different consequence rules, neither of the two can imply the other with the same $P$ and $Q$. For negated $P$ and $Q$, consider the simple program below
\begin{equation*}
	\mathsf{r1} \eqdef x := 1
\end{equation*}

	\noindent and the two sets of states $P_{\geq 0} \eqdef (x \geq 0)$ and $Q_{1} \eqdef (x = 1)$.
	Both the SIL triple $\angletriple{P_{\geq 0}}{\mathsf{r1}}{Q_{1}}$ and the IL triple $\undertriple{P_{\geq 0}}{\mathsf{r1}}{Q_{1}}$ are valid.
	However, neither $\undertriple{\lnot P_{\geq 0}}{\mathsf{r1}}{\lnot Q_{1}}$ nor $\angletriple{\lnot P_{\geq 0}}{\mathsf{r1}}{\lnot Q_{1}}$ are valid.
	So neither (IL) implies negated (SIL) nor the other way around.
	\qed
\end{example}

To gain some insights on why (SIL) and (IL) are not equivalent,
given a regular command $\regr$, we define the set of states that only diverges $D_{\regr}$ and the set of unreachable states $U_{\regr}$:
\[
D_{\regr} \eqdef \{ \sigma \svert \fwsem{\regr}\sigma= \emptyset \}
\qquad
U_{\regr} \eqdef \{ \sigma' \svert \sigma' \not\in\fwsem{\regr} \Sigma \} =  \{ \sigma' \svert \bwsem{\regr}\sigma' =\emptyset \}.
\]
In a sense, $U_{\regr}$ is the set of states which ``diverge" going backward.

\begin{lemma}\label{lmm:CC-1-monotone}
	For any regular command $\regr\in\Reg$ and sets of states $P,Q\subseteq\Sigma$ it holds that:
	\begin{enumerate}
		\item $\bwsem{\regr} \fwsem{\regr} P \supseteq P \setminus D_{\regr}$;
		\item $\fwsem{\regr} \bwsem{\regr} Q \supseteq Q \setminus U_{\regr}$.
	\end{enumerate}
\end{lemma}

Lemma~\ref{lmm:CC-1-monotone}
highlights the asymmetry between over and under-approximation: the composition of a function with its inverse is increasing (but for non-terminating states).
This explains why (HL) and (NC) are related while (IL) and (SIL) are not: on the over-approximating side, $P \setminus D_{\regr} \subseteq \bwsem{\regr} \fwsem{\regr} P$ can be further exploited if we know $\fwsem{\regr} P \subseteq Q$ via (HL), but it cannot when $\fwsem{\regr} P \supseteq Q$ via (IL). 

The preprint~\citep{unter2024} introduces a forward-oriented proof system with a core set of rules that are sound for both IL and SIL. It also becomes complete for IL, resp. SIL, when augmented with the corresponding consequence rule.

\subsection{Comparison Among the Rules}\label{sec:duality}

We recall the rules of SIL, HL and IL in Figure~\ref{fig:HLvsILvsSIL}, so to emphasize the similarities and differences among them.
The $\mathsf{atom}$ and $\mathsf{iter}$ rules show that HL and IL exploit the forward semantics, while SIL the backward one.
Furthermore, the rule \horule{iter} of HL says that any invariant is acceptable, not necessarily the minimal one, so that HL relies on over-approximation.
This is confirmed by the row for rules $\mathsf{cons}$ and $\mathsf{empty}$, where on the contrary IL and SIL are shown to rely on under-approximation.
The consequence rule is the key rule of all the logics because it allows to generalize a proof by weakening/strengthening the two conditions $P$ and $Q$ involved. The direction of rules \lrule{cons} of SIL and \horule{cons} of HL is the same and it is exactly the opposite of the direction of rule \irule{cons} of IL and NC, which coincides. So the different consequence rules follow the diagonals of Figure~\ref{fig:square-full}.
The rows for rules $\mathsf{seq}$ and $\mathsf{disj}$ show that in all cases triples can be composed sequentially and additively.
Rules $\mathsf{iter0}$, $\mathsf{unroll}$ and $\mathsf{unroll\mbox{-}split}$ are a prerogative of under-approximation: they are the same for SIL and IL, but they are unsound for HL.

\begin{figure}
	\resizebox{\textwidth}{!}{
		\begin{tabular}{c|ccc}
			Rule & SIL & HL & IL \\
			\hline \\
			$\mathsf{atom}$
			& \infer[]{\angletriple{\bwsem{\regc}Q}{\regc}{Q}}{}
			& \infer[]{\overtriple{P}{\regc}{\fwsem{\regc}P}}{}
			& \infer[]{\undertriple{P}{\regc}{\fwsem{\regc}P}}{}
			\\[1em]
			$\mathsf{cons}$
			& \quad\textcolor{ACMPurple}{\infer[]{\angletriple{P}{\regr}{Q}}{P \subseteq P' & \angletriple{P'}{\regr}{Q'}& Q' \subseteq Q}}\quad
			& \quad\textcolor{ACMPurple}{\infer[]{\overtriple{P}{\regr}{Q}}{P \subseteq P' & \overtriple{P'}{\regr}{Q'}& Q' \subseteq Q}}\quad
			& \quad\infer[]{\undertriple{P}{\regr}{Q}}{P \supseteq P' & \undertriple{P'}{\regr}{Q'}& Q' \supseteq Q}\quad
			\\[1em]
			$\mathsf{seq}$
			& \textcolor{ACMPurple}{\infer[]{\angletriple{P}{\regr_1;\regr_2}{Q}}{\angletriple{P}{\regr_1}{R} & \angletriple{R}{\regr_2}{Q}}}
			& \textcolor{ACMPurple}{\infer[]{\overtriple{P}{\regr_1;\regr_2}{Q}}{\overtriple{P}{\regr_1}{R} & \overtriple{R}{\regr_2}{Q}}}
			& \textcolor{ACMPurple}{\infer[]{\undertriple{P}{\regr_1;\regr_2}{Q}}{\undertriple{P}{\regr_1}{R} & \undertriple{R}{\regr_2}{Q}}}
			\\[1em]
			$\mathsf{choice}$
			& \infer[]{\angletriple{P_1 \cup P_2}{\regr_1 \regplus \regr_2}{Q}}{\forall i \in \{ 1, 2 \} & \angletriple{P_i}{\regr_i}{Q}}
			& \infer[]{\overtriple{P}{\regr_1 \regplus \regr_2}{Q}}{\forall i \in \{ 1, 2 \} & \overtriple{P}{\regr_i}{Q}}
			& \infer[]{\undertriple{P}{\regr_1 \regplus \regr_2}{Q_1 \cup Q_2}}{\forall i \in \{ 1, 2 \} & \undertriple{P}{\regr_i}{Q_i}}
			\\[1em]
			$\mathsf{iter}$
			& \infer[]{\angletriple{\bigcup\limits_{n \ge 0} Q_n}{\regr^\kstar}{Q_0}}{\forall n \ge 0 \sdot \angletriple{Q_{n+1}}{\regr}{Q_n}}
			& \infer[]{\overtriple{P}{\regr^\kstar}{P}}{\overtriple{P}{\regr}{P}}
			& \infer[]{\undertriple{P_0}{\regr^\kstar}{\bigcup\limits_{n \ge 0} P_n}}{\forall n \ge 0 \sdot \undertriple{P_n}{\regr}{P_{n+1}}}
			\\[1.5em]
			$\mathsf{empty}$
			& \textcolor{ACMPurple}{\infer[]{\angletriple{\emptyset}{\regr}{Q}}{}}
			& \textcolor{ACMPurple}{\infer[]{\overtriple{\emptyset}{\regr}{Q}}{}}
			& \infer[]{\undertriple{P}{\regr}{\emptyset}}{}
			\\[1em]
			$\mathsf{disj}$
			& \textcolor{ACMPurple}{\infer[]{\angletriple{P_1 \cup P_2}{\regr}{Q_1 \cup Q_2}}{\angletriple{P_1}{\regr}{Q_1} & \angletriple{P_2}{\regr}{Q_2} } }
			& \textcolor{ACMPurple}{\infer[]{\overtriple{P_1 \cup P_2}{\regr}{Q_1 \cup Q_2}}{\overtriple{P_1}{\regr}{Q_1} & \overtriple{P_2}{\regr}{Q_2} } }
			& \textcolor{ACMPurple}{\infer[]{\undertriple{P_1 \cup P_2}{\regr}{Q_1 \cup Q_2}}{\undertriple{P_1}{\regr}{Q_1} & \undertriple{P_2}{\regr}{Q_2} } }
			\\[1.5em]
			$\mathsf{iter0}$
			& \textcolor{ACMPurple}{\infer[]{\angletriple{Q}{\regr^\kstar}{Q}}{}}
			& unsound
			& \textcolor{ACMPurple}{\infer[]{\undertriple{P}{\regr^\kstar}{P}}{}}
			\\[1em]
			$\mathsf{unroll}$
			& \textcolor{ACMPurple}{\infer[]{\angletriple{P}{\regr^\kstar}{Q}}{\angletriple{P}{\regr^\kstar; \regr}{Q}}}
			& unsound
			& \textcolor{ACMPurple}{\infer[]{\undertriple{P}{\regr^\kstar}{Q}}{\undertriple{P}{\regr^\kstar; \regr}{Q}}}
			\\[1em]
			$\mathsf{unroll\mbox{-}split}$
			& \textcolor{ACMPurple}{\infer[]{\angletriple{P \cup Q_2}{\regr^\kstar}{Q_1 \cup Q_2}}{\angletriple{P}{\regr^\kstar; \regr}{Q_1}}}
			& unsound
			& \textcolor{ACMPurple}{\infer[]{\undertriple{P_1 \cup P_2}{\regr^\kstar}{Q \cup P_2}}{\undertriple{P_1}{\regr^\kstar; \regr}{Q}}}
		\end{tabular}
	}
	\caption{Comparison of SIL, HL and IL rules. We highlight in \textcolor{ACMPurple}{purple} when SIL rules are the same as HL or IL.}\label{fig:HLvsILvsSIL}
\end{figure}

\subsection{Weakest/Strongest Conditions}
Depending on the way in which program analysis is conducted, one can be interested in deriving either the most general or most specific hypotheses under which the reasoning can take place.
For instance, given a correctness program specification $Q$ one is typically interested in finding the weakest liberal preconditions that make $Q$ satisfied, i.e., to impose the minimal constraint on the input that guarantee program correctness.
Conversely, to infer necessary conditions we can be interested in devising the strongest hypotheses under which some correct run is possible. 

To investigate the existence of weakest/strongest pre and postconditions, it is convenient to take into account the consequence rules of the four kinds of triples.
The consequence rule of each logic explains how pre and postconditions can be weakened/strengthened.
The concrete semantics is trivially a strongest (HL and NC) or weakest (IL and SIL) condition for the target property ($P$ computing backward and $Q$ forward). 
It turns out that having a strongest/weakest condition on the ``source" property is a prerogative of over-approximation, i.e., that over and under-approximation are not fully dual theories.
\begin{proposition}[Existence of weakest conditions]\label{prop:weakest-pre-existence1}
	For any regular command $\regr\in\Reg$:
	\begin{itemize}
		\item given $Q$, there exists a weakest $P$ such that $\fwsem{\regr} P \subseteq Q$ (HL);
		\item given $P$, there exists a weakest $Q$ such that $\bwsem{\regr} Q \subseteq P$ (NC).
	\end{itemize}
\end{proposition}

\begin{proposition}[Non-existence of strongest conditions]\label{prop:weakest-pre-existence2}
	For any regular command $\regr\in\Reg$:
	\begin{itemize}
		\item for some $Q$, there is no strongest $P$ such that $\fwsem{\regr} P \supseteq Q$ (IL);
		\item for some $P$, there is no strongest $Q$ such that $\bwsem{\regr} Q \supseteq P$ (SIL).
	\end{itemize}
\end{proposition}

The reason why strongest conditions may not exist for IL and SIL is that the collecting semantics (both forward and backward) is additive but not co-additive. In other words, rule \horule{disj} is sound for all triples, while a dual rule for conjunction such as
\[
\infer[\{\mathsf{conj}\}]
{\overtriple{P_1 \cap P_2}{\regr}{Q_1 \cap Q_2}}
{\overtriple{P_1}{\regr}{Q_1} & \overtriple{P_2}{\regr}{Q_2} }
\]
is valid for HL and NC but neither for IL nor SIL.
So, for instance, given $Q$ and two HL preconditions $P_1$ and $P_2$ ($\overtriple{P_1}{\regr}{Q}$ and $\overtriple{P_2}{\regr}{Q}$) also their union is a precondition for $Q$, ie. $\overtriple{P_1 \cup P_2}{\regr}{Q}$, which can be proved using \horule{disj}. However, given two IL triples $\undertriple{P_1}{\regr}{Q}$ and $\undertriple{P_2}{\regr}{Q}$, in general $\fwsem{\regr} (P_1 \cap P_2) \nsupseteq Q$ in which case $\undertriple{P_1 \cap P_2}{\regr}{Q}$ is not valid.

\begin{example}\label{ex:il-no-strongest-pre}
	Consider again the program $\mathsf{r1}$ of Example~\ref{ex:il-sil-incomparable}.
	We can prove the two IL triples $\undertriple{x = 0}{\mathsf{r1}}{x = 1}$ and $\undertriple{x = 10}{\mathsf{r1}}{x = 1}$, but their intersection is $\undertriple{\emptyset}{\mathsf{r1}}{x = 1}$, which is not a valid IL triple.
	
	For SIL, consider the program
	\[
	\mathsf{rnd} \eqdef \code{x := nondet()}
	\]
	For precondition $P_1 \eqdef (x = 1)$ we can prove both $\angletriple{P_1}{\mathsf{rnd}}{x = 0}$ and $\angletriple{P_1}{\mathsf{rnd}}{x = 10}$, that are incomparable, and again are both minimal because $\emptyset$ is not a valid postcondition.
	\qed
\end{example}

\subsection{Termination and Reachability}

Termination and reachability are two sides of the same coin when switching from forward to backward reasoning, and over- and under-approximation behave differently with respect to this notion.

For HL, given the definition of collecting semantics, we can only distinguish a precondition which always causes divergence: if $Q$ is empty, all states in the precondition must always diverge. However, if just one state in $P$ has one terminating computation, its final state must be in $Q$, so we do not know any more whether states in $P$ diverge or not. Moreover, because of the over-approximation, a non empty $Q$ does not mean there truly are finite executions, as those may be introduced by the approximation.
Dually, NC cannot say much about reachability of $Q$ unless $P$ is empty, in which case $Q$ is unreachable.

On the contrary, under-approximation has much stronger guarantees on divergence/reachability. Any IL triple $\undertriple{P}{\regr}{Q}$ ensures that all states in $Q$ are reachable from states in $P$, which means in particular that every state in $Q$ is reachable. Dually, a SIL triple $\angletriple{P}{\regr}{Q}$ means that all states in $P$ have a convergent computation (which ends in a state in $Q$). This observation motivates the choice of a forward (resp. backward) rule for iteration in IL (resp. SIL): a backward (resp. forward) rule would need to prove reachability of all points in the postcondition (resp. precondition). Instead, the forward rule of IL (resp. backward rule of SIL) ensures reachability (resp. termination) by construction, as it builds $Q$ (resp. $P$) only with points which are known to be reachable (resp. terminating) by executing the loop.

\section{Separation Sufficient Incorrectness Logic}\label{sec:separation-sil}
We instantiate SIL to handle pointers and dynamic memory allocation, introducing Separation SIL. The goal of Separation SIL is to identify the causes of memory errors: it takes the backward under-approximation principles of SIL and combines it with the ability to deal with pointers from Separation Logic (SL) \citep{DBLP:conf/lics/Reynolds02,DBLP:conf/csl/OHearnRY01}.

\subsection{Heap Regular Commands}
We denote by $\Regh$ the set of all heap regular commands obtained by plugging the following definition of heap atomic commands in~(\ref{eq:reg-commands-def}) (in \textcolor{ACMBlue}{blue} the new primitives):
\[
\Cmdh \ni \regc ::= \; \code{skip} \mid \code{x := a} \mid \code{b?} \mid \textcolor{ACMBlue}{\code{x := alloc()} \mid \code{free(x)} \mid \code{x := [y]} \mid \code{[x] := y}} 
\]

The primitive \code{alloc()} allocates a new memory location containing a nondeterministic value, \code{free} deallocates memory and \code{[$\cdot$]} is the dereferencing operator. Please note that the syntax only allows to allocate, free and dereference (both for reading and writing) only single variables. Particularly, arithmetic $\code{a} \in \AExp$ and Boolean expressions $\code{b} \in \BExp$ cannot dereference a variable: to use a value from the heap, the value must be loaded in a variable beforehand.

Given a heap command $\regr \in \Regh$, we let $\fv(\regr) \subseteq \Var$ as the set of (free) variables appearing in $\regr$. We also define the set $\modified(\regr) \subseteq \Var$ of variables modified by $\regr$ inductively by
\begin{align*}
	&\modified(\code{skip}) = \emptyset \qquad &&\modified(\code{x := a}) = \{ \code{x} \} \\
	&\modified(\code{b?}) = \emptyset \qquad &&\modified(\code{x := alloc()}) = \{ \code{x} \} \\
	&\modified(\code{free(x)}) = \emptyset \qquad &&\modified(\code{x := [y]}) = \{ \code{x} \} \\
	&\modified(\code{[x] := y}) = \emptyset &&\modified(\regr_1; \regr_2) = \modified(\regr_1) \cup \modified(\regr_2) \\
	&\modified(\regr_1 \regplus \regr_2) = \modified(\regr_1) \cup \modified(\regr_2) &&\modified(\regr^{\kstar}) = \modified(\regr)
\end{align*}

Please note that \code{free(x)} and \code{[x] := y} do not modify \code{x}: this is because they only modify the value \emph{pointed by} \code{x}, not the actual value of \code{x} (the memory address itself).

\subsection{Assertion Language}
Our assertion language for pre and postconditions is derived from both SL and Incorrectness Separation Logic (ISL):
\[
\Asl \ni p, q, t ::= \false \mid \lnot p \mid p \land q \mid \exists x . p \mid a \asymp a \mid \emp \mid x \mapsto \code{a} \mid x \dealloc \mid p \andsep q
\]
In the above productions, $\asymp \in \{ =, \neq, \le, <, \dots \}$ encodes standard comparison operators, $x \in \Var$ is a generic variable and $\code{a} \in \AExp$ is an arithmetic expression. The first five constructs describe standard first order logic. The others describe heaps and come from Separation Logic, with the exception of $x \dealloc$, which was introduced by \citet{DBLP:conf/cav/RaadBDDOV20}. 

The constant $\emp$ denotes an empty heap.
The assertion $x \mapsto a$ stands for an heap with a single memory cell pointed by $x$ and whose content is $a$, while $x \dealloc$ describes that $x$ points to a memory cell that was previously deallocated. 
The separating conjunction $p \andsep q$ describes an heap which can be divided in two disjoint sub-heaps, one satisfying $p$ and the other $q$.
We let $x \mapsto - \eqdef \exists v. x \mapsto v$ describe that $x$ is allocated but we do not care about its exact value.
Given a formula $p \in \Asl$, we call $\fv(p) \subseteq \Var$ the set of its free variables.

\subsection{Proof System}\label{sec:sil-proof-system}

\begin{figure}[t]
	\centering
	\begin{framed}
	\hspace*{-0.6em}
	\(
	\begin{array}{cc}
		\infer[\lrule{skip}]
		{\angletriple{\emp}{\code{skip}}{\emp}}
		{}
		\; &
		\infer[\lrule{assign}]
		{\angletriple{q[a / x]}{\code{x := a}}{q}}
		{}
		\\[7.5pt]
		\infer[\lrule{assume}]
		{\angletriple{q \land b}{\code{b?}}{q}}
		{}
		\; &
		\infer[\lrule{alloc}]
		{\angletriple{\emp}{\code{x := alloc()}}{x \mapsto v}}
		{}
		\\[7.5pt]
		\infer[\lrule{free}]
		{\angletriple{x \mapsto -}{\code{free(x)}}{x \dealloc}}
		{}
		\; &
		\infer[\lrule{load}]
		{\angletriple{y \mapsto a \andsep q[a / x]}{\code{x := [y]}}{y \mapsto a \andsep q}}
		{x \notin \fv(a)}
		\\[7.5pt]
		\infer[\lrule{store}]
		{\angletriple{x \mapsto -}{\code{[x] := y}}{x \mapsto y}}
		{}
		\\[7.5pt]
		\hline\hline & \\[-2pt]
		\infer[\lrule{exists}]
		{\angletriple{\exists x. p}{\regr}{\exists x. q}}
		{\angletriple{p}{\regr}{q} & x \notin \fv(\regr)}
		\; &
		\infer[\lrule{frame}]
		{\angletriple{p \andsep t}{\regr}{q \andsep t}}
		{\angletriple{p}{\regr}{q} & \fv(t) \cap \modified(\regr) = \emptyset}
		\\[7.5pt]
		\hline\hline & \\[-2pt]
		\infer[\lrule{cons}]
		{\angletriple{p}{\regr}{q}}
		{p \Rightarrow p' & \angletriple{p'}{\regr}{q'}& q' \Rightarrow q}
		\; &
		\infer[\lrule{seq}]
		{\angletriple{p}{\regr_1; \regr_2}{q}}
		{\angletriple{p}{\regr_1}{t} & \angletriple{t}{\regr_2}{q}}
		\\[7.5pt]
		\infer[\lrule{choice}]
		{\angletriple{p_1 \lor p_2}{\regr_1 \regplus \regr_2}{q}}
		{\angletriple{p_1}{\regr_1}{q} & \angletriple{p_2}{\regr_2}{q}}
		\; &
		\infer[\lrule{iter}]
		{\angletriple{\exists n. q(n)}{\regr^\kstar}{q(0)}}
		{\forall n \ge 0 \;\; \angletriple{q(n+1)}{\regr}{q(n)}}
		\\[7.5pt]
		\hline\hline & \\[-2pt]
		\infer[\lrule{empty}]
		{\angletriple{\false}{\regr}{q}}
		{}
		\; &
		\infer[\lrule{disj}]
		{\angletriple{p_1 \lor p_2}{\regr}{q_1 \lor q_2}}
		{\angletriple{p_1}{\regr}{q_1} & \angletriple{p_2}{\regr}{q_2}}
		\\[7.5pt]
		\infer[\lrule{iter0}]
		{\angletriple{q}{\regr^{\kstar}}{q}}
		{}
		\; &
		\infer[\lrule{unroll}]
		{\angletriple{p}{\regr^{\kstar}}{q}}
		{\angletriple{p}{\regr^{\kstar}; \regr}{q}}
	\end{array}
	\)
	\end{framed}
	\caption{Proof rules for Separation SIL. The first group replaces SIL rule \lrule{atom}, the second includes rules peculiar of SL, the third includes rule from SIL core set, and the fourth includes additional SIL rules.}
	\label{fig:separation-sil}
\end{figure}

We present the rules of Separation SIL in Figure~\ref{fig:separation-sil}. We define the capture-avoiding substitution as usual: $q[a / x]$ is the formula obtained replacing all free occurrences of $x$ in $q$ with the expression $a$.
Unlike ISL, we do not distinguish between correct and erroneous termination -- the goal of SIL is to trace back the causes of errors, not to follow the flow of a program after an error has occurred.

We split the rules of Separation SIL in four groups (Figure~\ref{fig:separation-sil}). The first group includes new rules of Separation SIL. The second one includes rules borrowed from SL. The third and fourth ones are from SIL.

The first group gives the rules for atomic commands $\regc \in \Cmdh$, hence replacing the SIL rule \lrule{atom}.
Rule \lrule{skip} does not specify anything about its pre and postconditions, because whatever is true before and after the \code{skip} can be added with \lrule{frame}.
Rule \lrule{assign} is Hoare's backward assignment rule~\citep{hoare69}.
Rule \lrule{assume} conjoins the assertion \code{b} to the postcondition, because only states satisfying the Boolean guard have an execution.
Rule \lrule{alloc} allocates a new memory location for $x$. The premise is empty: if the previous content of $x$ is needed, \lrule{cons} can be used to introduce a constraint $x = x'$ in the premise. A frame can then reference $x'$ to talk about the previous content of $x$.
Rule \lrule{free} requires $x$ to be allocated before freeing it.
Rule \lrule{load} is similar to rule \lrule{assign}, with the addition of the (disjoint) $y \mapsto a$ to make sure that $y$ is allocated.
Rule \lrule{store} requires that $x$ is allocated, and updates the value it points to.
All these rules are local: they only specify pre and postconditions for the modified part of the heap, since anything else can be added via rule \lrule{frame}.

The rule \lrule{exists} allows to ``hide" local variables.
The rule \lrule{frame} is typical of separation logics \citep{DBLP:conf/lics/Reynolds02,DBLP:conf/cav/RaadBDDOV20}: it allows to add a frame around a derivation, plugging the proof for a small portion of a program inside a larger heap.
In the third group, we collected the core set presented in Figure~\ref{fig:sil}. The only notable difference is in rule \lrule{iter}, where Separation SIL uses a predicate $q(n)$ parametrized by the natural number $n \in \setN$ and the precondition $\exists n. q(n)$ in the conclusion of the rule. This is a logical replacement for the infinite union used in SIL rule \lrule{iter}.
In the fourth group, we instantiated the additional rules of Figure~\ref{fig:sil-extra}.

\subsection{Soundness}
\begin{figure}[t]
	\begin{subfigure}{\textwidth}
		\centering
		\begin{align*}
			\edenot{\code{skip}} (s, h) &\eqdef \{ (s, h) \} \\
			\edenot{\code{x := a}} (s, h) &\eqdef \left\lbrace (s[x \mapsto \edenot{\code{a}} s], h) \right\rbrace \\
			\edenot{\code{b?}} (s, h) &\eqdef \begin{cases*}
				\left\lbrace (s, h) \right\rbrace &if $\edenot{\code{b}} s = \code{tt}$\\
				\emptyset &otherwise
			\end{cases*} \\
			\edenot{\code{x := alloc()}} (s, h) &\eqdef \left\lbrace (s[x \mapsto l], h[l \mapsto v]) \svert v \in \Val, \mathit{avail}(l) \right\rbrace \\
			\edenot{\code{free(x)}} (s, h) &\eqdef \begin{cases*}
				\left\lbrace (s, h[s(x) \mapsto \bot]) \right\rbrace &if $h(s(x)) \in \Val$\\
				\left\lbrace \errstate \right\rbrace &otherwise
			\end{cases*} \\
			\edenot{\code{x := [y]}} (s, h) &\eqdef \begin{cases*}
				\left\lbrace (s[x \mapsto h(s(y))], h) \right\rbrace &if $h(s(y)) \in \Val$\\
				\left\lbrace \errstate \right\rbrace &otherwise
			\end{cases*} \\
			\edenot{\code{[x] := y}} (s, h) &\eqdef \begin{cases*}
				\left\lbrace (s, h[s(x) \mapsto s(y)]) \right\rbrace &if $h(s(x)) \in \Val$\\
				\left\lbrace \errstate \right\rbrace &otherwise
			\end{cases*}
		\end{align*}
		\caption{Semantics of heap atomic commands, where $\mathit{avail}(l) \eqdef (l \notin \dom(h) \lor h(l) = \bot)$.}
		\label{fig:ssil-model:commands}
	\end{subfigure}
	\begin{subfigure}{\textwidth}
		\centering
		\begin{align*}
			&\asldenot{\false} \eqdef \emptyset &&\asldenot{\lnot p} \eqdef \Sigma \setminus \asldenot{p} \\
			&\asldenot{p \land q} \eqdef \asldenot{p} \cap \asldenot{q} &&\asldenot{\exists x. p} \eqdef \{ (s, h) \svert \exists v \in \Val \sdot (s[x \mapsto v], h) \in \asldenot{p} \} \\
			&\asldenot{a_1 \asymp a_2} \eqdef \{ (s, h) \svert \edenot{a_1} s \asymp \edenot{a_2} s \} && \asldenot{\emp} \eqdef \{ (s, []) \} \\
			&\asldenot{x \mapsto a} \eqdef \{ (s, [s(x) \mapsto \edenot{a} s]) \} &&\asldenot{x \dealloc} \eqdef \{ (s, [s(x) \mapsto \bot]) \} \\
			&\rlap{$\asldenot{p \andsep q} \eqdef \{ (s, h_p \bullet h_q) \svert (s, h_p) \in \asldenot{p}, (s,h_q) \in \asldenot{q}, h_p \perp h_q \}$} &&
		\end{align*}
		\caption{Semantics of the assertion language.}
		\label{fig:ssil-model:assertions}
	\end{subfigure}
	\caption{All the ingredients to prove soundness of Separation SIL.}
\end{figure}

To prove soundness of Separation SIL, we first define a denotational semantics for heap regular commands.
We consider a finite set of variables $\Var$ and an infinite set $\Loc$ of memory locations; we define the set of values $\Val \eqdef \setZ \uplus \Loc$ where $\uplus$ is disjoint union.
Stores $s \in \Stores$ are (total) functions from variables to values (so either integers or memory locations). A heap $h \in \Heaps$ is a partial function $h: \Loc \rightharpoonup \Val \uplus \{ \bot \}$. If $h(l) = v \in \Val$, location $l$ is allocated and holds value $v$, if $l \notin \dom(h)$ then it is not allocated. The special value $\bot$ describes a deallocated memory location: if $h(l) = \bot$, that location was previously allocated and then deallocated.
As notation, we use $s[x \mapsto v]$ for function update. For heaps specifically, $[]$ is the empty heap, and $[l \mapsto v]$ is a shortcut for $[][l \mapsto v]$, that is the heap defined only on $l$ and associating to it value $v$.

We consider as states $\sigma$ pairs of a store and a heap, plus the special state \errstate{} representing the occurrence of an error: letting $\Sigma = \Stores \times \Heaps$, states are taken from $\Sigma_e = \Sigma \uplus \{ \errstate \}$. The denotational semantics of atomic commands $\edenot{\cdot}: \Cmdh \rightarrow \wp(\Sigma_e) \rightarrow \wp(\Sigma_e)$ is in Figure~\ref{fig:ssil-model:commands}. To simplify the presentation, we define it as $\edenot{\cdot}: \Cmdh \rightarrow \Sigma \rightarrow \wp(\Sigma_e)$, we let $\edenot{\regc} \errstate = \{ \errstate \}$, and we lift it to set of states by union.
Please note that, since arithmetic expressions $\code{a}$ and Boolean expressions $\code{b}$ cannot contain any dereferencing, their evaluation only depends on the store and not on the heap.

We define the forward collecting semantics of heap commands, $\fwsem{\cdot}: \Regh \rightarrow \wp(\Sigma_e) \rightarrow \wp(\Sigma_e)$, similarly to~(\ref{eq:fwsem-definition}) using the different semantics of atomic commands for $\regc \in \Cmdh$.

The semantics $\asldenot{\cdot}$ of a formula $p \in \Asl$ is a set of states in $\Sigma$.
As notation, we write $h_1 \perp h_2$ when $\dom(h_1) \cap \dom(h_2) = \emptyset$, and we say the two heaps are disjoint. For two disjoint heaps $h_1 \perp h_2$, we define the $\bullet$ operation as the merge of the two: $h_1 \bullet h_2$ is defined on $\dom(h_1) \cup \dom(h_2)$, and its value on $l$ is either $h_1(l)$ or $h_2(l)$ (the only one defined).
The full definition of $\asldenot{\cdot}$ is given in Figure~\ref{fig:ssil-model:assertions}.

Just like SIL, we define validity of a Separation SIL triple $\angletriple{p}{\regr}{q}$ by the condition $\bwsem{\regr} \asldenot{q} \supseteq \asldenot{p}$.
To prove soundness of Separation SIL, we rely on a stronger lemma, whose proof is by induction on the derivation tree.

\begin{lemma}\label{lmm:separation-sil-stronger-sound}
	Let $p, q, t \in \Asl$ and $\regr \in \Regh$. If $\angletriple{p}{\regr}{q}$ is provable and $\fv(t) \cap \modified(\regr) = \emptyset$ then: 
\[\bwsem{\regr} \asldenot{q \andsep t} \supseteq \asldenot{p \andsep t} .\]
\end{lemma}

Then we can prove soundness of the proof system by taking $t = \emp$ and using $p \andsep \emp \equiv p$.

\begin{corollary}[Separation SIL is sound]\label{th:separation-sil-sound}
	If a Separation SIL triple is provable then it is valid.
\end{corollary}

The presented proof system is not complete. However, we will show how to recover relative completeness~\citep[\S 4.3]{DBLP:journals/fac/AptO19} in Section~\ref{sec:separation-sil-completeness}.

\subsection{Example of SIL Derivation}\label{sec:separation-sil-derivation}
We now discuss in full detail the example in \citet{DBLP:conf/cav/RaadBDDOV20} (cf. Figure~\ref{fig:separationexample}) to show how Separation SIL can infer preconditions ensuring that a provided error can happen.

Our syntax does not support functions, so we assume \code{push\_back} to be an inlined macro.
As we cannot free and allocate $*v$ directly (just like ISL), we introduce an intermediate variable $y$. Thus we rewrite the program in Figure~\ref{fig:separationexample} as
\begin{equation*}
	\mathsf{rclient} \eqdef x := [v];\ (\regr_{b} \regplus \code{skip}) \qquad\qquad
	\regr_{b} \eqdef y := [v];\ \text{free}(y);\ y := \text{alloc}();\ [v] := y
\end{equation*}

\noindent To find the cause of errors, we do not include the last assignment \code{*x := 1} in $\mathsf{rclient}$: we know that whenever the postcondition $x \dealloc$ is satisfied, an error occurs after $\mathsf{rclient}$, and that is everything we need to find its source.
Any valid Separation SIL triple for $x \dealloc$ identifies a precondition such that any state satisfying it has a faulty execution.

\citet{DBLP:conf/cav/RaadBDDOV20} derives the Incorrectness Separation Logic triple below, which proves the existence of a faulty execution starting from at least one state in the precondition.
\[
\undertriple{v \mapsto z \andsep z \mapsto -}{\mathsf{rclient}}{v \mapsto y \andsep y \mapsto - \andsep x \dealloc}.
\]

Separation SIL can do more: it can prove the triple
\[
\angletriple{v \mapsto z \andsep z \mapsto - \andsep \true}{\mathsf{rclient}}{x \dealloc \andsep \true}
\]
which has both a more succinct postcondition capturing the error and a stronger guarantee: \emph{every} state in the precondition reaches the error, hence it gives (many) actual witnesses for testing and debugging purposes.
Moreover, Separation SIL proof system guides the crafting of the precondition if the proof is done from the error postcondition (e.g., the pointer deallocated right before its dereference) backward.

\begin{figure}[t]
	\begin{subfigure}[b]{\textwidth}
		\footnotesize
		\begin{align*}
			&\angleexact{\true \andsep v \mapsto z \andsep z \mapsto - \andsep (x = z \lor x \dealloc)} \\
			&\quad y := [v]; \\
			& \textcolor{gray}{\angleexact{\true \andsep \underline{v \mapsto z \andsep y \mapsto - \andsep (x = y \lor x \dealloc)}}} \\
			&\angleexact{\true \andsep v \mapsto - \andsep y \mapsto - \andsep (x = y \lor x \dealloc)} \\
			&\quad \text{free}(y); \\
			& \textcolor{gray}{\angleexact{\true \andsep v \mapsto - \andsep \underline{y \dealloc} \andsep (x = y \lor x \dealloc)}}\\
			&\angleexact{x \dealloc \andsep v \mapsto - \andsep \emp \andsep \true} \\
			&\quad y := \text{alloc}(); \\
			& \textcolor{gray}{\angleexact{x \dealloc \andsep v \mapsto - \andsep \underline{y \mapsto y'} \andsep \true}} \\
			&\angleexact{x \dealloc \andsep v \mapsto - \andsep \true} \\
			&\quad [v] := y \\
			& \textcolor{gray}{\angleexact{x \dealloc \andsep \underline{v \mapsto y} \andsep \true}} \\
			&\angleexact{x \dealloc \andsep \true}
		\end{align*}
		\caption{Derivation of the triple $\angletriple{t}{\regr_{b}}{q}$, linearized. We write in \textcolor{gray}{grey} the strengthened conditions obtained using \lrule{cons}, and \underline{underline} the postcondition of the rule for the current atomic command. Everything else is a frame shared between pre and post, using \lrule{frame}.}
		\label{fig:ssil-derivation:sub1}
	\end{subfigure}
	\begin{subfigure}[b]{\textwidth}
		\centering
		\footnotesize
		\[
		\infer[\lrule{seq}]
		{\angletriple{p}{\mathsf{rclient}}{q}}
		{
			\infer[\lrule{cons}]{\angletriple{p}{x := [v]}{t \lor q}}
			{
				\infer[\lrule{load}]{\angletriple{p}{x := [v]}{t}}{}
			}
			&
			\infer[\lrule{choice}]{\angletriple{t \lor q}{\regr_{b} \regplus \code{skip}}{q}}{
				\infer[]{\angletriple{t}{\regr_{b}}{q}}{\vdots} &
				\infer[\lrule{frame}]{\angletriple{q}{\code{skip}}{q}}{
					\infer[\lrule{skip}]{\angletriple{\emp}{\code{skip}}{\emp}}{}
				}
			}
		}
		\]
		\caption{Derivation of the Separation SIL triple $\angletriple{p}{\mathsf{rclient}}{q}$ exploiting the sub-derivation in Figure~\ref{fig:ssil-derivation:sub1}.}
		\label{fig:ssil-derivation:sub2}
	\end{subfigure}
	\caption{The full derivation of $\angletriple{p}{\mathsf{rclient}}{q}$, split in two parts.}
	\label{fig:ssil-derivation}
\end{figure}

Let us fix the following assertions:
\begin{align*}
p: &(v \mapsto z \andsep z \mapsto - \andsep \true), \\
q: & (x \dealloc \andsep \true), \\
t: & (v \mapsto z \andsep z \mapsto - \andsep (x = z \lor x \dealloc) \andsep \true).
\end{align*}

To prove the Separation SIL triple $\angletriple{p}{\mathsf{rclient}}{q}$, we first prove $\angletriple{t}{\regr_b}{q}$. The derivation of this triple is in Figure~\ref{fig:ssil-derivation:sub1}.
This derivation is better understood if read from bottom to top, since we start from the postcondition and look for a suitable precondition to apply the rule for each one of the four atomic commands. In all cases, we start with a postcondition, then strengthen it to be able to apply the right rule: this usually means adding some constraint on the shape of the heap.
In particular, to apply the rule \lrule{free} we need $y$ to be deallocated, and this can happen in two different ways: either if $y = x$, since $x$ is deallocated; or if $y$ is a new name. This is captured by the disjunction $x = y \lor x \dealloc$.

Using the derivation in Figure~\ref{fig:ssil-derivation:sub1}, we complete the proof as shown in Figure~\ref{fig:ssil-derivation:sub2}.
We can apply \lrule{load} to prove the triple $\angletriple{p}{\code{x := [v]}}{t}$ because $p$ is equivalent to 
\[
(v \mapsto z \andsep z \mapsto - \andsep (z = z \lor z \dealloc) \andsep \true):
\]

\noindent $z \mapsto - \andsep z \dealloc$ is not satisfiable, so we can remove that disjunct.

The same example was used in~\citet{DBLP:journals/corr/ZilbersteinDS23} to illustrate the effectiveness of outcome-based separation logic for bug-finding. Even though the OL derivation shown in~\citet[Fig.~6]{DBLP:journals/corr/ZilbersteinDS23} proves essentially the same triple as the SIL one in Fig.~\ref{fig:ssil-derivation}, the deduction processes are quite different.
In fact, OL reasoning is forward oriented, as witnessed by the presence of the implication that concludes the proof and by the triple for the $\code{skip}$ branch, whereas SIL is naturally backward oriented, to infer the preconditions that lead to the error.

\subsection{Observations on SIL Principles}\label{sec:separation-sil-principles}
In Section~\ref{sec:separation-sil-derivation}, we use \lrule{cons} to drop the disjunct $q$ from $\angletriple{p}{\code{x := [v]}}{t \lor q}$. Similarly, we could have used \lrule{cons} to drop the disjunct $x \dealloc$ in the precondition for $y := \text{alloc}()$ in $\regr_b$. This is analogous to the IL ability to drop disjuncts in the post, but with respect to the backward direction.

Furthermore, we use the postcondition $x \dealloc \andsep \true$. The reader might be wondering why we had to include the $(\andsep\ \true)$: is it not possible to just frame it in when we plug the proof in a larger program?
The issue is that in final reachable states $x$ is not the only variable allocated (there are also $v$ and $y$), so the final heap should talk about them as well. Adding $(\andsep\ \true)$ is just a convenient way to focus only on the part of the heap that describes the error, that is $x \dealloc$, and just leave everything else unspecified since we do not care about it.

\subsection{Relative completeness of Separation SIL}\label{sec:separation-sil-completeness}
The proof system in Section~\ref{sec:sil-proof-system} is not complete. To move towards completeness, we first limit the assertion language to the existential fragment of first-order logic:
\begin{align*}
	\Asl \ni p, q, t ::= &\; \false \mid \true \mid p \land q \mid p \lor q \mid \exists x . p \mid a \asymp a \\
	&\mid \emp \mid x \mapsto \code{a} \mid x \dealloc \mid p \andsep q
\end{align*}
We remove negation, so we don't include universal quantifiers and heap assertions must be positive. However, we argue that this is sufficient to find bugs: for instance, in the example in Section~\ref{sec:separation-sil-derivation}, we only used assertions from this subset.

With this limited assertion language, the proof system in Section~\ref{sec:sil-proof-system} is complete for all atomic commands except \code{alloc}.
To deal with \code{alloc}, we need the ability to refer to the specific memory location that was allocated. However, the naive solution to add a constraint $x = \alpha$ in the post of \lrule{alloc} makes the frame rule unsound: for instance, the following triple is not valid:
\[
\angletriple{\emp \andsep \alpha \mapsto -}{\code{x := alloc()}}{(x \mapsto - \land x = \alpha) \andsep \alpha \mapsto -} .
\]
To recover the frame rule, just like ISL needs the deallocated assertion in the post \citep[§3]{DBLP:conf/cav/RaadBDDOV20}, we need a "will be allocated" assertion in the pre. To this end we use the $\dealloc$ assertion, and change the semantic model to only allocate a memory location that is explicitly $\bot$ instead of one not in the domain of the heap.
We formalize this by letting $\mathit{avail}(l) \eqdef h(l) = \bot$ in Figure~\ref{fig:ssil-model:commands}, and replacing the axiom \lrule{alloc} with 
\[
\infer[\lrule{alloc}]
{\angletriple{\beta \dealloc{}}{\code{x := alloc()}}{x = \beta \land x \mapsto v}}
{}
\]
Soundness still holds for this different semantics. Moreover, we can prove relative completeness \citep[§4.3]{DBLP:journals/fac/AptO19} for loop-free programs:
\begin{theorem}[Relative completeness for loop-free programs]\label{th:separation-sil-sequential-complete}
	Suppose to have an oracle to prove implications between formulas in $\Asl$. Let $\regr \in \Regh$ be a regular command without $\kstar$ and $p, q \in \Asl$ such that $\bwsem{\regr} \asldenot{q} \supseteq \asldenot{p}$. Then the triple $\angletriple{p}{\regr}{q}$ is provable.
\end{theorem}

The proof relies on the possibility to rewrite any $q$ in an equivalent assertion of the form $\exists x_1.\cdots\exists x_n.~\bigvee_{1\leq i\leq k} q_i$ where all $q_i$ are assertions involving atoms composed with $\wedge$ and $\andsep$ only.
This way, completeness is proved for such $q_i$ first and then extended to the entire $q$ thanks to rules \lrule{disj} and \lrule{exists}. 
Notably, we show that the weakest (possible) precondition $\bwsem{\regr}\asldenot{q}$ of loop-free programs is always expressible as an assertion $t \in \Asl$, namely $\asldenot{t} = \bwsem{\regr} \asldenot{q}$, and prove that the triple $\angletriple{t}{\regr}{q}$ can be derived. Then, by \lrule{cons}, the theorem follows for any $p$ that implies $t$.

\section{Conclusion and Future Work}\label{sec:conc}

We have introduced SIL as a correct and complete program logic aimed to locate the causes of errors.
Furthermore, we instantiated SIL to a sound and (relatively) complete proof system for handling memory errors and discussed its advantages over ISL and OL.

Unlike IL, which was designed to expose erroneous outputs, SIL provides sufficient conditions that explain why such errors can occur.
SIL can be characterized as a logic based on backward under-approximation, which helped us to compare it against HL, IL and NC.
This is captured in the taxonomy of Figure~\ref{fig:square-full} and in the rule-by-rule comparison of Figure~\ref{fig:HLvsILvsSIL}, which we used to clarify the analogies and differences between the possible approaches.
We obtained some surprising connections: although NC and IL share the same consequence rules, they are not comparable; NC triples are isomorphic to HL ones, but such a correspondence cannot be extended to relate IL and SIL; and we pointed out the main reasons why duality arguments cannot apply in this case.
The following list addresses the overall connections between the different logics:
\begin{description}
\item[HL vs NC:] there is an isomorphism given by $\overtriple{P}{\regr}{Q}$ iff $\necctriple{\neg P}{\regr}{\neg Q}$.
\item[HL vs IL:] in general there is no relation between forward over- and under\hyp{}approximation triples: the only triples common to HL and IL are the exact ones.
\item[HL vs SIL:] for deterministic and terminating programs HL and SIL judgements do coincide. More precisely, we have that for terminating programs $\overtriple{P}{\regr}{Q}$ implies $\angletriple{P}{\regr}{Q}$ and for deterministic programs $\angletriple{P}{\regr}{Q}$ implies $\overtriple{P}{\regr}{Q}$. OL~\citep{DBLP:journals/corr/ZilbersteinDS23}  is correct for both HL and SIL and the revised OL in the preprint~\citep{zilberstein2024relatively} is also complete.
\item[NC vs IL:] although NC and IL share the same consequence rule, they are not comparable unless the program is reversible (in which case IL implies NC).
\item[NC vs SIL:] in general there is no relation between backward over- and under\hyp{}approximation triples: the only triples common to NC and SIL are the exact ones.
\item[IL vs SIL:] there is no relation. The independently developed preprint \citep{unter2024} proposes a correct and complete proof system for IL and SIL.
\end{description}

We conclude that each logic focuses on different aspects: none of them is ``better'' than the others; instead, each one has its merits.

\subsection{Future Work}

It seems interesting to explore a synergic use of different logics for practical issues. For instance, a forward, IL-based analysis determines, at every program point, a set of truly reachable states, which can then be used to narrow down a subsequent backward, SIL-based analysis to find input states which truly generate the pointed errors.
To this aim, it could be useful to extend the taxonomy of HL, NC, IL and SIL by incorporating some other approaches~\citep{DBLP:conf/ecoop/MaksimovicCLSG23,DBLP:journals/corr/ZilbersteinDS23,unter2024,DBLP:journals/jacm/Hoare78}, as discussed in Section~\ref{sec:relatedwork}.

Lastly, we plan to explore the expressiveness of other assertion languages in order to extend the relative completeness proof of Separation SIL to deal with loops.

\begin{acks}
This work is supported by the \grantsponsor{1}{Italian Ministero dell'Università e della Ricerca}{https://prin.mur.gov.it/} under Grant No. \grantnum{1}{P2022HXNSC}, PRIN 2022 PNRR -- \emph{Resource Awareness in Programming: Algebra, Rewriting, and Analysis}.
\end{acks}

\bibliographystyle{ACM-Reference-Format}
\bibliography{bibfile}

\appendix

\newpage
\section{Proofs}\label{sec:proofs}

\subsection{Proofs of Section~\ref{sec:background}}

\begin{proof}[Proof of Lemma~\ref{lmm:bwsem-calculus}]
	In the proof, we assume $Q$ to be any set of states, and $\sigma' \in Q$ to be any of its elements.
	
	\proofcase{$\bwsem{\regr_1; \regr_2}$}
	By~(\ref{eq:bwsem-sigma-sigma'}), $\sigma \in \bwsem{\regr_1; \regr_2} \sigma'$ if and only if $\sigma' \in \fwsem{\regr_1; \regr_2} \sigma$.
	\[
	\fwsem{\regr_1; \regr_2} \sigma = \fwsem{\regr_2} (\fwsem{\regr_1} \sigma) = \bigcup\limits_{\sigma'' \in \fwsem{\regr_1} \sigma} \fwsem{\regr_2} \sigma''
	\]
	so $\sigma' \in \fwsem{\regr_1; \regr_2} \sigma$ if and only if there exists a $\sigma'' \in \fwsem{\regr_1} \sigma$ such that $\sigma' \in \fwsem{\regr_2} \sigma''$. Again by~(\ref{eq:bwsem-sigma-sigma'}), these are equivalent to $\sigma \in \bwsem{\regr_1} \sigma''$ and $\sigma'' \in \bwsem{\regr_2} \sigma'$, respectively. Hence
	\[
	\sigma' \in \fwsem{\regr_1; \regr_2} \sigma \iff \exists \sigma'' \in \bwsem{\regr_2} \sigma' \sdot \sigma \in \bwsem{\regr_1} \sigma''
	\]
	Since $\bwsem{\cdot}$ is defined on sets by union
	\[
	\bwsem{\regr_1} (\bwsem{\regr_2} \sigma') = \bigcup\limits_{\sigma'' \in \bwsem{\regr_2} \sigma'} \bwsem{\regr_1} \sigma''
	\]
	which means $\exists \sigma'' \in \bwsem{\regr_2} \sigma' \sdot \sigma \in \bwsem{\regr_1} \sigma''$ if and only if $\sigma \in \bwsem{\regr_1} (\bwsem{\regr_2} \sigma')$.
	Putting everything together, we get $\sigma \in \bwsem{\regr_1; \regr_2} \sigma'$ if and only if $\sigma \in \bwsem{\regr_1} (\bwsem{\regr_2} \sigma')$, so the two are the same set. The thesis follows easily lifting the equality by union on $\sigma' \in Q$ and by the arbitrariness of $Q$.
	
	\proofcase{$\bwsem{\regr_1 \regplus \regr_2}$}
	By~(\ref{eq:bwsem-sigma-sigma'}), $\sigma \in \bwsem{\regr_1 \regplus \regr_2} \sigma'$ if and only if $\sigma' \in \fwsem{\regr_1 \regplus \regr_2} \sigma$.
	\[
	\fwsem{\regr_1 \regplus \regr_2} \sigma = \fwsem{\regr_1} \sigma \cup \fwsem{\regr_2} \sigma
	\]
	so $\sigma' \in \fwsem{\regr_1 \regplus \regr_2} \sigma$ if and only if $\exists i \in \{ 1, 2 \}$ such that $\sigma' \in \fwsem{\regr_i} \sigma$. This is again equivalent to $\sigma \in \bwsem{\regr_i} \sigma'$, and
	\[
	\exists i \in \{ 1, 2 \} \sdot \sigma \in \bwsem{\regr_i} \sigma' \iff \sigma \in \bwsem{\regr_1} \sigma' \cup \bwsem{\regr_2} \sigma'
	\]
	Putting everything together, we get $\sigma \in \bwsem{\regr_1 \regplus \regr_2} \sigma'$ if and only if $\sigma \in \bwsem{\regr_1} \sigma' \cup \bwsem{\regr_2} \sigma'$, which implies the thesis as in point 1.
	
	\proofcase{$\bwsem{{\regr^\kstar}}$}
	To prove this last equality, we define $\regr^n$ inductively as the sequential composition of $\regr$ with itself n times: $\regr^1 = \regr$ and $\regr^{n+1} = \regr^n; \regr$. Clearly $\fwsem{\regr^n} =\fwsem{\regr}^n$. For simplicity, we also define $\fwsem{\regr^0} = \bwsem{\regr^0} = \fwsem{\regr}^0$. We prove by induction on $n$ that $\bwsem{\regr^n} = \bwsem{\regr}^n$.
	For $n = 1$ we have $\bwsem{\regr^1} = \bwsem{\regr}^1$. If we assume it holds for $n$ we have
	
	\begin{align*}
		\bwsem{\regr^{n+1}} &= \bwsem{\regr^n; \regr} &[\text{def. of }\regr^n] \\
		&= \bwsem{\regr^n} \circ \bwsem{\regr} &[\text{pt. 1 of this lemma}] \\
		&= \bwsem{\regr}^n \circ \bwsem{\regr} &[\text{inductive hp}] \\
		&= \bwsem{\regr}^{n+1}
	\end{align*}
	
	We then observe that
	\begin{align*}
		\bwsem{\regr^\kstar} \sigma' &= \{ \sigma \svert \sigma' \in \fwsem{\regr^\kstar} \sigma \} &[\text{def. of } \bwsem{\cdot}] \\
		&= \{ \sigma \svert \sigma' \in \bigcup\limits_{n \ge 0} \fwsem{\regr}^n \sigma \} &[\text{def. of } \fwsem{\regr^\kstar}] \\
		&= \bigcup\limits_{n \ge 0} \{ \sigma \svert \sigma' \in \fwsem{\regr}^n \sigma \} & \\
		&= \bigcup\limits_{n \ge 0} \{ \sigma \svert \sigma' \in \fwsem{\regr^n} \sigma \} &[\text{observed above}]\\
		&= \bigcup\limits_{n \ge 0} \{ \sigma \svert \sigma \in \bwsem{\regr^n} \sigma' \} &[(\ref{eq:bwsem-sigma-sigma'})]\\
		&= \bigcup\limits_{n \ge 0} \bwsem{\regr^n} \sigma' \\
		&= \bigcup\limits_{n \ge 0} \bwsem{\regr}^n \sigma' &[\text{shown above}]
	\end{align*}
	
	As in the cases above, the thesis follows.
\end{proof}

\subsection{Proofs of Section~\ref{sec:comparison}}

\begin{proof}[Proof of Proposition~\ref{prop:fw-inclusion-negation-bw}]
	We prove the left-to-right implication, so assume $\fwsem{\regr}P \subseteq Q$.
	Take a state $\sigma' \in \lnot Q$. This means $\sigma' \notin Q$, that implies $\sigma' \notin \fwsem{\regr} P$. So, for any state $\sigma \in P$, we have $\sigma' \notin \fwsem{\regr} \sigma$, which is equivalent to $\sigma \notin \bwsem{\regr} \sigma'$ by~(\ref{eq:bwsem-sigma-sigma'}). This being true for all $\sigma \in P$ means $P \cap \bwsem{\regr} \sigma' = \emptyset$, that is equivalent to $\bwsem{\regr} \sigma' \subseteq \lnot P$.
	Since this holds for all states $\sigma' \in \lnot Q$, we have $\bwsem{\regr} (\lnot Q) \subseteq \lnot P$.
	
	The other implication is analogous.
\end{proof}

\begin{proof}[Proof of Proposition~\ref{prop:sil-hl-deterministic-terminating}]
	To prove the first point, assume $\bwsem{\regr} Q \supseteq P$ and take $\sigma' \in \fwsem{\regr} P$. Then there exists $\sigma \in P$ such that $\sigma' \in \fwsem{\regr} \sigma$. Since $\regr$ is deterministic, $\fwsem{\regr} \sigma$ can contain at most one element, hence $\fwsem{\regr} \sigma = \{ \sigma' \}$. Moreover, since $\sigma \in P \subseteq \bwsem{\regr} Q$ there must exists a $\sigma'' \in Q$ such that $\sigma'' \in \fwsem{\regr} \sigma = \{ \sigma' \}$, which means $\sigma' \in Q$. Again, by arbitrariness of $\sigma' \in \fwsem{\regr} P$, this implies $\fwsem{\regr} P \subseteq Q$.

	To prove the second point, assume $\fwsem{\regr} P \subseteq Q$ and take a state $\sigma \in P$. Since $\regr$ is terminating, $\fwsem{\regr} \sigma$ is not empty, hence we can take $\sigma' \in \fwsem{\regr} \sigma$. The hypothesis $\fwsem{\regr} P \subseteq Q$ implies that $\sigma' \in Q$. Then, by~(\ref{eq:bwsem-sigma-sigma'}), $\sigma \in \bwsem{\regr} \sigma' \subseteq \bwsem{\regr} Q$. By arbitrariness of $\sigma \in P$, this implies $P \subseteq \bwsem{\regr} Q$.
\end{proof}

\begin{proof}[Proof of Lemma~\ref{lmm:CC-1-monotone}]
	We first prove that $\bwsem{\regr} \fwsem{\regr} P \supseteq P \setminus D_{\regr}$.
	Take a $\sigma \in P \setminus D_{\regr}$. Because $\sigma \notin D_{\regr}$, $\fwsem{\regr} \sigma \neq \emptyset$, so take $\sigma' \in \fwsem{\regr} \sigma$. Since $\sigma \in P$ we have $\sigma' \in \fwsem{\regr} P$. Moreover, by~(\ref{eq:bwsem-sigma-sigma'}), we get $\sigma \in \bwsem{\regr} \sigma' \subseteq \fwsem{\regr} P$. By arbitrariness of $\sigma \in P \setminus D_{\regr}$ we have the thesis.
	
	The proof for $\fwsem{\regr} \bwsem{\regr} Q \supseteq Q \setminus U_{\regr}$ is analogous.
\end{proof}

\begin{proof}[Proof of Proposition~\ref{prop:weakest-pre-existence1}]
	By definition, $\fwsem{\regr}$ is additive, that is $\fwsem{\regr}(P_1 \cup P_2) = \fwsem{\regr} P_1 \cup \fwsem{\regr} P_2$. Take all $P$ such that $\fwsem{\regr} P \subseteq Q$. By additivity of $\fwsem{\regr}$, their union satisfies the same inequality, hence it is the weakest such $P$.

	By definition, $\bwsem{\regr}$ is additive. Analogously, take all $Q$ such that $\bwsem{\regr} Q \subseteq P$. By additivity of $\bwsem{\regr}$, their union is the weakest $Q$ satisfying that inequality.
\end{proof}

\begin{proof}[Proof of Proposition~\ref{prop:weakest-pre-existence2}]
	The proof is given by the counterexamples in Example~\ref{ex:il-no-strongest-pre}.
	For IL, the example shows that for $Q_{1} \eqdef (x = 1)$ there is no strongest $P$ such that $\fwsem{\mathsf{r1}} P \supseteq Q$: $x = 0$ and $x = 10$ are incomparable and are both minimal, as $\emptyset$ is not a valid precondition.
	The argument for SIL is analogous using $\mathsf{rnd}$ and precondition $P_{1} \eqdef (x = 1)$.
\end{proof}

\subsection{SIL Soundness and Completeness}

We split the proof between soundness and completeness.
\begin{proposition}[SIL is sound]\label{prop:sil-correct}
	Any provable SIL triple is valid.
\end{proposition}
\begin{proof}
	The proof is by structural induction on the derivation tree, using the characterization of Proposition~\ref{prop:sil-validity-characterization}.
\end{proof}

\begin{proposition}[SIL is complete]\label{prop:sil-complete}
	Any valid SIL triple is provable.
\end{proposition}
\begin{proof}
	First we show that, for any $Q$, the triple $\angletriple{\bwsem{\regr}Q}{\regr}{Q}$ is provable by induction on the structure of $\regr$.

	\proofcase{$\regr = \regc$}
	We can prove $\angletriple{\bwsem{\regc}Q}{\regc}{Q}$ using \lrule{atom}.

	\proofcase{$\regr = \regr_1; \regr_2$}
	We can prove $\angletriple{\bwsem{\regr}Q}{\regr_1; \regr_2}{Q}$ with
	\[
		\infer[\lrule{seq}]
		{\angletriple{\bwsem{\regr_1}\bwsem{\regr_2}Q}{\regr_1; \regr_2}{Q}}
		{\angletriple{\bwsem{\regr_1}\bwsem{\regr_2}Q}{\regr_1}{\bwsem{\regr_2}Q} & \angletriple{\bwsem{\regr_2}Q}{\regr_2}{Q}}
	\]
	where the two premises can be proved by inductive hypothesis, and $\bwsem{\regr_1; \regr_2}Q = \bwsem{\regr_1}\bwsem{\regr_2}Q$ by Lemma~\ref{lmm:bwsem-calculus}.
	
	\proofcase{$\regr = \regr_1 \regplus \regr_2$}
	We can prove $\angletriple{\bwsem{\regr}Q}{\regr_1 \regplus \regr_2}{Q}$ with
	\[
		\infer[\lrule{choice}]
		{\angletriple{\bwsem{\regr_1} Q \cup \bwsem{\regr_2}Q}{\regr_1 \regplus \regr_2}{Q}}
		{\forall i \in \{ 1, 2 \} & \angletriple{\bwsem{\regr_i}Q}{\regr_i}{Q}}
	\]
	where the two premises can be proved by inductive hypothesis, and $\bwsem{\regr_1 \regplus \regr_2} Q = \bwsem{\regr_1} Q \cup \bwsem{\regr_2}Q$ by Lemma~\ref{lmm:bwsem-calculus}.
	
	\proofcase{$\regr = \regr^\kstar$}
	We can prove $\angletriple{\bwsem{\regr^\kstar}Q}{\regr^\kstar}{Q}$ with
	\[
		\infer[\lrule{iter}]
		{\angletriple{\bigcup\limits_{n \ge 0} \bwsem{\regr}^n Q}{\regr}{Q}}
		{\forall n \ge 0 \sdot \angletriple{\bwsem{\regr}^{n+1} Q}{\regr}{\bwsem{\regr}^{n} Q}}
	\]
	where the premises can be proved by inductive hypothesis since $\bwsem{\regr}^{n+1} Q = \bwsem{\regr} \bwsem{\regr}^{n} Q$, and $\bwsem{\regr^{\kstar}}Q = \bigcup\limits_{n \ge 0} \bwsem{\regr}^n Q$ by Lemma~\ref{lmm:bwsem-calculus}.
	
	To conclude the proof, take a triple $\angletriple{P}{\regr}{Q}$ such that $\bwsem{\regr} Q \supseteq P$. Then we can first prove the triple $\angletriple{\bwsem{\regr}Q}{\regr}{Q}$, and then using rule \lrule{cons} we derive $\angletriple{P}{\regr}{Q}$.
\end{proof}

The proof of Theorem~\ref{thm:sil-sound-complete} is a corollary of Proposition~\ref{prop:sil-correct}--\ref{prop:sil-complete}.

\subsection{Other Proofs about SIL}
\begin{proof}[Proof of Proposition~\ref{prop:sil-validity-characterization}]
	By definition of $\bwsem{\cdot}$ we have
	\[
	\bwsem{\regr} Q = \bigcup_{\sigma' \in Q} \{ \sigma \svert \sigma \in \bwsem{\regr} \sigma' \} = \{ \sigma \svert \exists \sigma' \in Q \sdot \sigma' \in \fwsem{\regr} \sigma \}
	\]

	\noindent Using this,
	\[
	P \subseteq \bwsem{\regr} Q \iff \forall \sigma \in P \sdot \sigma \in \{ \sigma \svert \exists \sigma' \in Q \sdot \sigma' \in \fwsem{\regr} \sigma \} \iff \forall \sigma \in P \sdot \exists \sigma' \in Q \sdot \sigma' \in \fwsem{\regr} \sigma
	\]
\end{proof}

\begin{proposition}[Validity of additional SIL rules]
	The rules in Figure~\ref{fig:sil-extra} are correct, that is triples provable in SIL extended with those rules are valid.
\end{proposition}
\begin{proof}
	The proof is by structural induction on the derivation tree, extends that of Proposition~\ref{prop:sil-correct} with inductive cases for the new rules and relies on the same characterization from Proposition~\ref{prop:sil-validity-characterization}.
\end{proof}

\subsection{Proofs about Separation SIL}\label{sec:proofs-ssil}
Given two stores $s, s' \in \Stores$ and a heap command $\regr \in \Regh$, we use the notation $s \dotsim_{\regr} s'$ to indicate that they coincide on all variables not modified by $\regr$: $\forall x \notin \modified(\regr) \sdot s(x) = s'(x)$. Please note that $\dotsim_{\regr}$ is an equivalence relation.

\begin{lemma}\label{lmm:store-only-change-mod}
	Let $(s, h) \in \Sigma$, $\regr \in \Regh$. If $(s', h') \in \fwsem{\regr}(s, h)$ then $s \dotsim_{\regr} s'$.
\end{lemma}
\begin{proof}
	The proof is by induction on the syntax of $\regr$. We prove here only some relevant cases.
	
	\proofcase{\code{x := a}}
	$(s', h') \in \fwsem{\code{x := a}}(s, h)$ means that $s' = s[x \mapsto \edenot{\code{a}} s]$. Particularly, this means that for all variables $y \neq x$, $s'(y) = s(y)$, which is the thesis because $\modified(\code{x := a}) = \{ x \}$.
	
	\proofcase{\code{free(x)}}
	$(s', h') \in \fwsem{\code{free(x)}}(s, h)$ means that $s' = s$, which is the thesis because $\modified(\code{free(x)}) = \emptyset$.
	
	\proofcase{$\regr_1; \regr_2$}
	$(s', h') \in \fwsem{\regr_1; \regr_2}(s, h)$ means that there exists $(s'', h'') \in \fwsem{\regr_1} (s, h)$ such that $(s', h') \in \fwsem{\regr_2}(s'', h'')$. By inductive hypothesis, since $\modified(\regr_1) \subseteq \modified(\regr_1; \regr_2)$, we have $s'' \dotsim_{\regr_1; \regr_2} s$. Analogously, $\modified(\regr_2) \subseteq \modified(\regr_1; \regr_2)$ implies $s' \dotsim_{\regr_1; \regr_2} s''$. From these, we get $s' \dotsim_{\regr_1; \regr_2} s$.
\end{proof}

The next technical proposition states some semantic properties of the assertion language to be exploited in the proof of Lemma~\ref{lmm:separation-sil-stronger-sound}.

\begin{proposition}\label{prop:sl-sat}
	Let $p \in \Asl$, $s, s' \in \Stores$, $h \in \Heaps$ and $a \in \text{AExp}$.
	\begin{enumerate}
		\item If $\forall x \in \fv(p) \sdot s(x) = s'(x)$ and $(s, h) \in \asldenot{p}$ then $(s', h) \in \asldenot{p}$.
		\item If $(s, h) \in \asldenot{p[a / x]}$ then $(s[x \mapsto \edenot{a} s], h) \in \asldenot{p}$.
	\end{enumerate}
\end{proposition}

\begin{proof}
	By structural induction on the syntax of assertions.
\end{proof}

\begin{proof}[Proof of Lemma~\ref{lmm:separation-sil-stronger-sound}]
	First, we observe that Proposition~\ref{prop:sil-validity-characterization} does not depend on the specific definition of $\fwsem{\cdot}$, thus it holds for separation SIL as well. Thanks to this, we prove the thesis through the equivalent condition
	\[
	\forall (s, h) \in \asldenot{p \andsep t} \sdot \exists (s', h') \in \asldenot{q \andsep t} \sdot (s', h') \in \fwsem{\regr} (s, h)
	\]

	The proof is by induction on the derivation tree of the provable triple $\angletriple{p}{\regr}{q}$. We prove here only some relevant cases.

	\proofcase{\lrule{assign}}
	Take $(s, h) \in \asldenot{q[a / x] \andsep t}$. Then we can split $h = h_p \bullet h_t$ such that $(s, h_p) \in \asldenot{q[a / x]}$ and $(s, h_t) \in \asldenot{t}$. Let $s' = s[x \mapsto \edenot{\code{a}} s]$, so that $(s', h) \in \fwsem{\code{x := a}} \asldenot{q[a / x] \andsep t}$. Since $\fv(t) \cap \modified(\regr) = \emptyset$, $x \notin \fv(t)$. Thus, by Proposition~\ref{prop:sl-sat}.1, $(s', h_t) \in \asldenot{t}$.
	Moreover, $(s', h_p) \in \asldenot{q}$ by Proposition~\ref{prop:sl-sat}.2.
	Hence, $(s', h_p \bullet h_t) = (s', h) \in \asldenot{q \andsep t}$.

	\proofcase{\lrule{alloc}}
	Take $(s, h) \in \asldenot{\emp \andsep t} = \asldenot{t}$. Take a location $l \notin \dom(h)$, and let $s' = s[x \mapsto l]$, $h' = h[l \mapsto s(v)]$, so that $(s', h') \in \fwsem{\code{x := alloc()}} (s, h)$.
	We can split $h' = [l \mapsto s(v)] \bullet h$ because $l \notin \dom(h)$. Since $\fv(t) \cap \modified(\regr) = \emptyset$, $x \notin \fv(t)$. Thus, by Proposition~\ref{prop:sl-sat}.1, $(s', h) \in \asldenot{t}$.
	Moreover, $(s', [l \mapsto s(v)]) = (s', [s'(x) \mapsto s'(v)])$, which satisfies $(s', [s'(x) \mapsto s'(v)]) \in \asldenot{x \mapsto v}$.
	Hence $(s', h') \in \asldenot{x \mapsto v \andsep t}$.

	\proofcase{\lrule{load}}
	Take $(s, h) \in \asldenot{y \mapsto a \andsep q[a /x] \andsep t}$. Then we know $x \notin \fv(t)$ and $h = [s(y) \mapsto \edenot{a} s] \bullet h_p \bullet h_t$, $(s, h_p) \in \asldenot{q[a / x]}$, $(s, h_t) \in \asldenot{t}$.
	Let $s' = s[x \mapsto h(s(y))] = s[x \mapsto \edenot{a} s]$.
	By Proposition~\ref{prop:sl-sat}.1, $(s', h_t) \in \asldenot{t}$.
	By Proposition~\ref{prop:sl-sat}.2, $(s', h_p) \in \asldenot{q}$.
	Lastly, since $x \notin \fv(a)$, $\edenot{a} s' = \edenot{a} s$ and $s'(y) = s(y)$, so we have $(s', [s'(y) \mapsto \edenot{a} s']) \in \asldenot{y \mapsto a}$.
	Combining these, $(s', h) = (s', [s(y) \mapsto \edenot{a} s] \bullet h_p \bullet h_t) \in \asldenot{y \mapsto a \andsep q \andsep t}$.
	The thesis follows observing that $(s', h) \in \fwsem{\code{x := [y]}} (s, h)$.
	
	\proofcase{\lrule{store}}
	Take $(s, h) \in \asldenot{x \mapsto - \andsep t}$. Then $x \notin \fv(t)$ and exists $v \in \Val$ such that $h = [s(x) \mapsto v] \bullet h_t$, $(s, h_t) \in \asldenot{t}$.
	Let $h' = h[s(x) \mapsto s(y)]$. Clearly $h' = [s(x) \mapsto s(y)] \bullet h_t$ and $(s, [s(x) \mapsto s(y)]) \in \asldenot{x \mapsto y}$.
	Hence $(s, h') \in \asldenot{x \mapsto y \andsep t}$ and $(s, h') \in \fwsem{[x] := y} (s, h)$, which is the thesis.
	
	\proofcase{\lrule{exists}}
	Take $(s, h) \in \asldenot{(\exists x . p) \andsep t}$. Then there exists a value $v \in \Val$ and decomposition $h = h_p \bullet h_t$ such that $(s[x \mapsto v], h_p) \in \asldenot{p}$ and $(s, h_t) \in \asldenot{t}$. Without loss of generality, we can assume $x \notin \fv(t)$; otherwise, we just rename it using a fresh name neither in $t$ nor in $\regr$. Hence, by Proposition~\ref{prop:sl-sat}.1, $(s[x \mapsto v], h_t) \in \asldenot{t}$. So $(s[x \mapsto v], h) \in \asldenot{p \andsep t}$.
	By inductive hypothesis on the provable triple $\angletriple{p}{\regr}{q}$ and formula $t$, there is $(s', h') \in \asldenot{q \andsep t}$ such that $(s', h') \in \fwsem{\regr} (s[x \mapsto v], h)$. Because $x \notin \fv(\regr)$, we also have $(s', h') \in \fwsem{\regr} (s, h)$, and clearly $(s', h') \in \asldenot{(\exists x . q) \andsep t}$, that is the thesis.
	
	\proofcase{\lrule{frame}}
	Take $(s, h) \in \asldenot{p \andsep t' \andsep t}$. By hypothesis, $(\fv(t' \andsep t)) \cap \modified(\regr) = (\fv(t') \cup \fv(t)) \cap \modified(\regr) = \emptyset$. Then, applying the inductive hypothesis on the provable triple $\angletriple{p}{\regr}{q}$ and the formula $t' \andsep t$ (which satisfies the hypothesis of the theorem) we get exactly the thesis.
	
	\proofcase{\lrule{seq}}
	Because of name clashes, here we assume the hypotheses of rule \lrule{seq} to be $\angletriple{p}{\regr_1}{p'}$ and $\angletriple{p'}{\regr_2}{q}$.
	Since $\modified(r_1) \cup \modified(r_2) = \modified(r_1; r_2)$, we know that $\fv(t) \cap \modified(r_1) = \fv(t) \cap \modified(r_2) = \emptyset$.
	Take $(s, h) \in \asldenot{p \andsep t}$. By inductive hypothesis on provable triple $\angletriple{p}{\regr_1}{p'}$ and formula $t$ we get that there exists $(s'', h'') \in \asldenot{p' \andsep t}$ such that $(s'', h'') \in \fwsem{\regr_1} (s, h)$. Then, by inductive hypothesis on the provable triple $\angletriple{p'}{\regr_2}{q}$ and formula $t$ again, we get $(s', h') \in \asldenot{q \andsep t}$ such that $(s', h') \in \fwsem{\regr_2} (s'', h'')$.
	The thesis follows since $(s', h') \in \fwsem{\regr_2} (s'', h'') \subseteq \fwsem{\regr_2} (\fwsem{\regr_1} (s, h)) = \fwsem{\regr_1; \regr_2} (s, h)$.
\end{proof}

\subsection{Completeness of Separation SIL}\label{sec:proofs-ssil-completeness}
Some of the following equivalences are standard, but we collect them all here for convenience.
\begin{lemma}\label{lmm:asl-equivalences}
	For all assertions $p_1, p_2, q$, variables $x, x' \in \Var$ and arithmetic expressions $a_1, a_2 \in \AExp$, the following equivalences hold:
	\begin{enumerate}
		\item $(p_1 \lor p_2) \land q \equiv (p_1 \land q) \lor (p_2 \land q)$
		\item $(p_1 \lor p_2) \andsep q \equiv (p_1 \andsep q) \lor (p_2 \andsep q)$
		\item $\exists x. (p \lor q) \equiv (\exists x. p) \lor (\exists x. q)$
		\item $\exists x. (p \land q) \equiv (\exists x. p) \land q$ \quad if $x \notin \fv(q)$
		\item $\exists x. (p \andsep q) \equiv (\exists x. p) \andsep q$ \quad if $x \notin \fv(q)$
		\item $a_1 \asymp a_2 \land (p_1 \andsep p_2) \equiv (a_1 \asymp a_2 \land p_1) \andsep p_2$
		\item $(p_1 \andsep x \mapsto x') \land (p_2 \andsep x \mapsto x') \equiv (p_1 \land p_2) \andsep x \mapsto x'$
		\item $(p_1 \andsep x \dealloc) \land (p_2 \andsep x \dealloc) \equiv (p_1 \land p_2) \andsep x \dealloc$
	\end{enumerate}
\end{lemma}
\begin{proof}
	Point (1), (3) and (4) are standard in first-order logic. Point (2), (5) and (6) are standard in separation logic \citep{DBLP:conf/lics/Reynolds02}.

	For point (7) we observe
	\begin{align*}
		&\asldenot{(p_1 \andsep x \mapsto x') \land (p_2 \andsep x \mapsto x')} \\
		&=\{ (s, h) \svert (s, h) \in \asldenot{p_1 \andsep x \mapsto x'}, (s, h) \in \asldenot{p_2 \andsep x \mapsto x'} \} \\
		&=\lbrace (s, h) \svert h = h_1 \bullet [s(x) \mapsto s(x')], (s, h_1) \in \asldenot{p_1} h = h_2 \bullet [s(x) \mapsto s(x')], (s, h_2) \in \asldenot{p_2} \rbrace \\
		&=\{ (s, h) \svert h = h' \bullet [s(x) \mapsto s(x')], (s, h') \in \asldenot{p_1}, (s, h') \in \asldenot{p_2}  \} \\
		&=\asldenot{x \mapsto x' \andsep (p_1 \land p_2)}
	\end{align*}

	Point (8) is analogous.
\end{proof}

\begin{lemma}\label{lmm:separation-assertion-rewrite}
	Let $q \in \Asl$ be a formula without $\lor$ and $\exists$, and let $x'$ be a fresh variable. Then,
	\begin{enumerate}
		\item for some $p$, $q \land (x \mapsto x' \andsep \true) \equiv x \mapsto x' \andsep p$
		\item for some $p$, $q \land (x \dealloc \andsep \true) \equiv x \dealloc \andsep p$
	\end{enumerate}
\end{lemma}
\begin{proof}
	The proof is by induction on the structure of $q$.

	\proofcase{$q = \false$}
	Take $p = \false$.

	\proofcase{$q = \true$}
	Take $p = \true$.

	\proofcase{$q = q_1 \land q_2$}
	We consider point (1) first. By inductive hypothesis, there exists $p_1$ and $p_2$ such that
	\begin{align*}
		q_1 \land (x \mapsto x' \andsep \true) &\equiv x \mapsto x' \andsep p_1 \\
		q_2 \land (x \mapsto x' \andsep \true) &\equiv x \mapsto x' \andsep p_2
	\end{align*}
	so that
	\begin{align*}
		q \land (x \mapsto x' \andsep \true) &\equiv q_1 \land q_2 \land (x \mapsto x' \andsep \true) \\
		&\equiv q_1 \land (x \mapsto x' \andsep \true) \land q_2 \land (x \mapsto x' \andsep \true) \\
		&\equiv (x \mapsto x' \andsep p_1) \land (x \mapsto x' \andsep p_2) \\
		&\equiv x \mapsto x' \andsep (p_1 \land p_2)
	\end{align*}
	where we used Lemma~\ref{lmm:asl-equivalences}.7 for the last equivalence.
	The case for point (2) is analogous using Lemma~\ref{lmm:asl-equivalences}.8 instead.

	\proofcase{$q = a_1 \asymp a_2$}
	Both points follow from Lemma~\ref{lmm:asl-equivalences}.6 by taking $p_1 = \true$ and $p_2 = x \mapsto x'$ (resp. $p_2 = x \dealloc$).

	\proofcase{$q = \emp$}
	Both formulas $\emp \land (x \mapsto x' \andsep \true)$ and $\emp \land (x \dealloc \andsep \true)$ are not satisfiable. Therefore we get the thesis with $p = \false$.

	\proofcase{$q = z \mapsto z'$}
	For point (1):
	\begin{align*}
		&\asldenot{z \mapsto z' \land (x \mapsto x' \andsep \true)} \\
		&=\{ (s, h) \svert (s, h) \in \asldenot{z \mapsto z'}, (s, h) \in \asldenot{x \mapsto x' \andsep \true} \} \\
		&=\{ (s, h) \svert h = [s(z) \mapsto s(z')], h = [s(x) \mapsto s(x')] \bullet h_t \} \\
		&=\{ (s, h) \svert h = [s(x) \mapsto s(x')], s(z) = s(x), s(z') = s(x') \} \\
		&=\asldenot{z' = x' \land z = x \land x \mapsto x'}
	\end{align*}

	For point (2), we observe that $z \mapsto z' \land (x \dealloc \andsep \true)$ is not satisfiable, so we get the thesis with $p = \false$.

	\proofcase{$q = z \dealloc$}
	For point (1), we observe that $z \dealloc{} \land (x \mapsto x' \andsep \true)$ is not satisfiable, so we get the thesis with $p = \false$.

	For point (2):
	\begin{align*}
		&\asldenot{z \dealloc \land (x \dealloc \andsep \true)} \\
		&=\{ (s, h) \svert (s, h) \in \asldenot{z \dealloc}, (s, h) \in \asldenot{x \dealloc \andsep \true} \} \\
		&=\{ (s, h) \svert h = [s(z) \mapsto \bot], h = [s(x) \mapsto \bot] \bullet h_t \} \\
		&=\{ (s, h) \svert h = [s(x) \mapsto \bot], s(z) = s(x) \} \\
		&=\asldenot{z = x \land x \dealloc}
	\end{align*}

	\proofcase{$q = q_1 \andsep q_2$}
	We consider point (1) first. By inductive hypothesis, there exists $p_1$ and $p_2$ such that
	\begin{align*}
		q_1 \land (x \mapsto x' \andsep \true) &\equiv x \mapsto x' \andsep p_1 \\
		q_2 \land (x \mapsto x' \andsep \true) &\equiv x \mapsto x' \andsep p_2
	\end{align*}
	Take $p = p_1 \andsep q_2 \lor q_1 \andsep p_2$. We have
	\begin{align*}
		\asldenot{x \mapsto x' \andsep p} = \asldenot{x \mapsto x' \andsep p_1 \andsep q_2} \cup \asldenot{x \mapsto x' \andsep q_1 \andsep p_2}
	\end{align*}
	Now consider
	\begin{align*}
		&\asldenot{q \land (x \mapsto x' \andsep \true)} \\
		&=\asldenot{(q_1 \andsep q_2) \land (x \mapsto x' \andsep \true)} \\
		&=\lbrace (s, h) \svert (s, h) \in \asldenot{x \mapsto x' \andsep \true}, h = h_1 \bullet h_2, (s, h_1) \in \asldenot{q_1}, (s, h_2) \in \asldenot{q_2} \rbrace \\
		&=\lbrace (s, h) \svert h(s(x)) = s(x'), h = h_1 \bullet h_2, (s, h_1) \in \asldenot{q_1}, (s, h_2) \in \asldenot{q_2} \rbrace
	\end{align*}
	For every state $(s, h)$ in this set, either $s(x) \in \dom(h_1)$ or $s(x) \in \dom(h_2)$: it can't be in neither because $h(s(x)) = s(x')$. Consider the former case: then $(s, h_1) \in \asldenot{x \mapsto x' \andsep \true}$, so that $(s, h_1) \in \asldenot{q_1 \land (x \mapsto x' \andsep \true)} = \asldenot{x \mapsto x' \andsep p_1}$, so that $(s, h_1 \bullet h_2) \in \asldenot{x \mapsto x' \andsep p_1 \andsep q_2}$. Analogously, in the latter case $(s, h_1 \bullet h_2) \in \asldenot{x \mapsto x' \andsep q_1 \andsep p_2}$.
	Therefore, $\asldenot{q \land (x \mapsto x' \andsep \true)} \subseteq \asldenot{x \mapsto x' \andsep p}$.

	For the other inclusion, consider
	\begin{align*}
		&\asldenot{x \mapsto x' \andsep p_1 \andsep q_2} \\
		&= \lbrace (s, h) \svert h = h_1 \bullet h_2, (s, h_1) \in \asldenot{x \mapsto x' \andsep p_1}, (s, h_2) \in \asldenot{q_2} \rbrace \\
		&= \lbrace (s, h) \svert h = h_1 \bullet h_2, (s, h_1) \in \asldenot{q_1 \land (x \mapsto x' \andsep \true)}, (s, h_2) \in \asldenot{q_2} \rbrace \\
		&= \lbrace (s, h) \svert h = h_1 \bullet h_2, h_1(s(x)) = s(x'), (s, h_1) \in \asldenot{q_1}, (s, h_2) \in \asldenot{q_2} \rbrace \\
		&\subseteq \lbrace (s, h) \svert h = h_1 \bullet h_2, h(s(x)) = s(x'), (s, h_1) \in \asldenot{q_1}, (s, h_2) \in \asldenot{q_2} \rbrace \\
		&= \asldenot{q \land (x \mapsto x' \andsep \true)}
	\end{align*}
	and analogously for $\asldenot{x \mapsto x' \andsep q_1 \andsep p_2} \subseteq \asldenot{q \land (x \mapsto x' \andsep \true)}$.

	The case for point (2) is analogous.
\end{proof}

\begin{lemma}\label{lmm:separation-sil-atom-base-completeness}
	Let $q \in \Asl$ be an assertion without $\lor$ and $\exists$, and let $\regc \in \Cmdh$. Then there exists $p \in \Asl$ such that $\asldenot{p} = \bwsem{\regc} \asldenot{q}$, and $\angletriple{p}{\regc}{q}$ is provable.
\end{lemma}
\begin{proof}
	We recall that
	\[
	\bwsem{\regc} \asldenot{q} = \{ (s, h) \svert \edenot{\regc} (s, h) \cap \asldenot{q} \neq \emptyset \}
	\]
	and that $\errstate \notin \asldenot{q}$ for any $q$.
	In the proof below, we will use the following equivalence: given a state $(s, h)$ such that $s(h(x)) \in \Val$, $(s, h) \in \asldenot{x \mapsto - \andsep \true}$. Therefore, $(s, h) \in \asldenot{q}$ if and only if $(s, h) \in \asldenot{q \land (x \mapsto - \andsep \true)}$. Using Lemma~\ref{lmm:separation-assertion-rewrite}.1, there exists a $q'$ (which depends on $q$ and $x$ but not on $(s, h)$) such that this is true if and only if $(s, h) \in \asldenot{\exists x'. (x \mapsto x' \andsep q')}$.
	Analogously (using Lemma~\ref{lmm:separation-assertion-rewrite}.2), if $s(h(x)) = \bot$, $(s, h) \in \asldenot{q}$ if and only if $(s, h) \in \asldenot{x \dealloc \andsep q'}$ for some $q'$.

	We now proceed by cases on the heap atomic command $\regc$.

	\proofcase{\code{skip}}
	We have
	\begin{align*}
		\bwsem{\regc} \asldenot{q} = \{ (s, h) \svert \edenot{\code{skip}} (s, h) \cap \asldenot{q} \neq \emptyset \} = \{ (s, h) \svert (s, h) \in \asldenot{q} \}
	\end{align*}
	So we have the thesis taking $p = q$, and we prove it by using \lrule{skip} and \lrule{frame}.

	\proofcase{\code{x := a}}
	We have
	\begin{align*}
		\bwsem{\regc} \asldenot{q} &= \{ (s, h) \svert \edenot{\code{x := a}} (s, h) \cap \asldenot{q} \neq \emptyset \} \\
		&= \{ (s, h) \svert (s[x \mapsto \edenot{\code{a}} s], h) \in \asldenot{q} \} \\
		&= \{ (s, h) \svert (s, h) \in \asldenot{q[a / x]} \}
	\end{align*}
	So we have the thesis taking $p = q[a / x]$, and we prove it by using \lrule{assign}.

	\proofcase{\code{b?}}
	We have
	\begin{align*}
		\bwsem{\regc} \asldenot{q} &= \{ (s, h) \svert \edenot{\code{b?}} (s, h) \cap \asldenot{q} \neq \emptyset \} \\
		&= \{ (s, h) \svert \edenot{\code{b}} s = \code{tt}, (s, h) \in \asldenot{q} \} \\
		&= \asldenot{q \land b}
	\end{align*}
	So we have the thesis taking $p = q \land b$, and we prove it by using \lrule{assume}.

	\proofcase{\code{x := alloc()}}
	We have
	\begin{align*}
		\bwsem{\regc} \asldenot{q} &= \{ (s, h) \svert \edenot{\code{x := alloc()}} (s, h) \cap \asldenot{q} \neq \emptyset \} \\
		&= \{ (s, h) \svert \exists l, v. h(l) = \bot, (s[x \mapsto l], h[l \mapsto v]) \in \asldenot{q} \} \\
		&= \lbrace (s, h) \svert \exists l, v. h(l) = \bot, (s[x \mapsto l], h[l \mapsto v]) \in \asldenot{\exists x'. x \mapsto x' \andsep q'} \rbrace \\
		&= \lbrace (s, h) \svert \exists l, v. h(l) = \bot, \exists v'. (s[x \mapsto l][x' \mapsto v'], h[l \mapsto v]) \in \asldenot{x \mapsto x' \andsep q'} \rbrace \\
		&= \lbrace (s, h) \svert \exists l, v. h = [l \mapsto \bot] \bullet h_q, \exists v'. \\
		&\rlap{$ \qquad\qquad\quad (s[x \mapsto l][x' \mapsto v'], [l \mapsto v]) \in \asldenot{x \mapsto x'}, $} \\
		&\rlap{$ \qquad\qquad\quad (s[x \mapsto l][x' \mapsto v'], h_q) \in \asldenot{q'} \rbrace$} \\
		&= \lbrace (s, h) \svert \exists l, v. h = [l \mapsto \bot] \bullet h_q, (s[x \mapsto l][x' \mapsto v], h_q) \in \asldenot{q'} \rbrace \\
		&= \asldenot{\exists i. \exists x'. i \dealloc \andsep q'[i / x]}
	\end{align*}
	for fresh variables $i$. So we have the thesis taking $p = \exists i. \exists x'. i \dealloc \andsep q'[i / x]$. To prove the triple $\angletriple{p}{\code{x := alloc()}}{q}$ we first observe the following chain of implications:
	\begin{align*}
		q &\impliedby \exists x'. x \mapsto x' \andsep q' &[\text{Lemma~\ref{lmm:separation-assertion-rewrite}.1}]\\
		&\quad\,\equiv\; \exists i. \exists x'. x \mapsto x' \andsep q' &[i \text{ fresh}] \\
		&\impliedby \exists i. \exists x'. x = i \land (x \mapsto x' \andsep q') & \\
		&\impliedby \exists i. \exists x'. x = i \land (x \mapsto x' \andsep q'[i / x]) &[\text{replacing } i = x] \\
		&\impliedby \exists i. \exists x'. (x = i \land x \mapsto x') \andsep q'[i / x] &[\text{Lemma~\ref{lmm:asl-equivalences}.6}]
	\end{align*}
	Then we prove the triple with the following derivation tree:
	\[
	\infer[\lrule{cons}]{\angletriple{p}{\regc}{q}}{
		\infer[\lrule{exists} \text{ x2}]{\angletriple{p}{\regc}{\exists i. \exists x'. (x = i \land x \mapsto x') \andsep q'[i / x]}}{
			\infer[\lrule{frame}]{\angletriple{i \dealloc \andsep q'[i / x]}{\regc}{(x = i \land x \mapsto x') \andsep q'[i / x]}}{
				x \notin \fv(q'[i / x])
				&\infer[\lrule{alloc}]{ \angletriple{i \dealloc}{\regc}{x = i \land x \mapsto x'} }{}
			}
		}
	}
	\]

	\proofcase{\code{free(x)}}
	We have
	\begin{align*}
		\bwsem{\regc} \asldenot{q} &= \{ (s, h) \svert \edenot{\code{free(x)}} (s, h) \cap \asldenot{q} \neq \emptyset \} \\
		&= \{ (s, h) \svert h(s(x)) \in \Val, (s, h[s(x) \mapsto \bot]) \in \asldenot{q} \} \\
		&= \{ (s, h) \svert h(s(x)) \in \Val, (s, h[s(x) \mapsto \bot]) \in \asldenot{x \dealloc \andsep q'} \} \\
		&= \lbrace (s, h) \svert h(s(x)) \in \Val, h = [s(x) \mapsto h(s(x))] \bullet h_q, (s, h_q) \in \asldenot{q'} \rbrace \\
		&= \asldenot{x \mapsto - \andsep q'}
	\end{align*}
	So we have the thesis taking $p = x \mapsto - \andsep q'$, and we prove it by using \lrule{free} and \lrule{frame} with frame $q'$ (this is always possible because $\modified(\code{free(x)}) = \emptyset$).

	\proofcase{\code{x := [y]}}
	We have
	\begin{align*}
		\bwsem{\regc} \asldenot{q} &= \{ (s, h) \svert \edenot{\code{x := [y]}} (s, h) \cap \asldenot{q} \neq \emptyset \} \\
		&= \{ (s, h) \svert h(s(y)) \in \Val, (s[x \mapsto h(s(y))], h) \in \asldenot{q} \} \\
		&= \lbrace (s, h) \svert h(s(y)) \in \Val, (s[x \mapsto h(s(y))], h) \in \asldenot{\exists y'. y \mapsto y' \andsep q'} \rbrace \\
		&= \lbrace (s, h) \svert h(s(y)) \in \Val, h = [s(y) \mapsto h(s(y))] \bullet h_q, (s[x \mapsto h(s(y))], h_q) \in \asldenot{q'} \rbrace \\
		&= \asldenot{\exists y'. (y \mapsto y' \andsep q'[y' / x])}
	\end{align*}
	So we have the thesis taking $p = \exists y'. (y \mapsto y' \andsep q'[y' / x])$, and we prove it by using \lrule{load} with $a = y'$ and \lrule{exists} because $y'$ is fresh.

	\proofcase{\code{[x] := y}} We have
	\begin{align*}
		\bwsem{\regc} \asldenot{q} &= \{ (s, h) \svert \edenot{\code{[x] := y}} (s, h) \cap \asldenot{q} \neq \emptyset \} \\
		&= \{ (s, h) \svert h(s(x)) \in \Val, (s, h[s(x) \mapsto s(y)]) \in \asldenot{q} \} \\
		&= \lbrace (s, h) \svert h(s(x)) \in \Val, (s, h[s(x) \mapsto s(y)]) \in \asldenot{\exists x'. x \mapsto x' \andsep q'} \rbrace \\
		&= \lbrace (s, h) \svert h(s(x)) \in \Val, h = [s(x) \mapsto h(s(x))] \bullet h_q, (s, h_q) \in \asldenot{q'} \rbrace \\
		&= \asldenot{x \mapsto - \andsep q'}
	\end{align*}
	So we have the thesis taking $p = x \mapsto - \andsep q'$. To prove the triple $\angletriple{p}{\code{[x] := y}}{q}$, we first prove $\angletriple{p}{\code{[x] := y}}{x \mapsto y \andsep q'}$ by using \lrule{store} and \lrule{frame} with frame $q'$ (this is always possible because $\modified(\code{[x] := y}) = \emptyset$). Then we observe that $x \mapsto y \implies \exists x'. x \mapsto x'$, so that we get $\angletriple{p}{\code{[x] := y}}{q}$ by using \lrule{cons}.
\end{proof}

\begin{proof}[Proof of Theorem~\ref{th:separation-sil-sequential-complete}]
	First we fix $q$ and prove, by induction on the structure of $\regr$, that there exists $p \in \Asl$ such that $\asldenot{p} = \bwsem{\regr} \asldenot{q}$, and $\angletriple{p}{\regr}{q}$ is provable.

	\proofcase{$\regr = \regc$}
	First, we transform $q$ in a normal form: we rename all quantified variables to fresh names, and use Lemma~\ref{lmm:asl-equivalences} (points 1-5) to lift disjunctions to the top, then existential quantifiers. Thus, without loss of generality, we can assume that $q$ is a disjunction of existentially quantified formulas that don't contain $\lor$ and $\exists$. Moreover, if we have a proof for each one of these formulas without $\lor$ and $\exists$, we can combine them using rules \lrule{disj} and \lrule{exists} to get a proof for the original $q$. Therefore, again without loss of generality, we can consider only the case in which $q$ does not contain $\lor$ and $\exists$.
	This case is exactly Lemma~\ref{lmm:separation-sil-atom-base-completeness}, so we conclude the inductive step.

	\proofcase{$\regr = \regr_1; \regr_2$}
	By inductive hypothesis on $\regr_2$, we know that there exists an assertion $t \in \Asl$ such that $\asldenot{t} = \bwsem{\regr_2} \asldenot{q}$ and $\angletriple{t}{\regr_2}{q}$ is provable. By inductive hypothesis on $\regr_1$, we know that there exists an assertion $p \in \Asl$ such that $\asldenot{p} = \bwsem{\regr_1} \asldenot{t}$ and $\angletriple{p}{\regr_1}{t}$ is provable. Now $\asldenot{p} = \bwsem{\regr_1} \asldenot{t} = \bwsem{\regr_1} \bwsem{\regr_2} \asldenot{q} = \bwsem{\regr_1; \regr_2} \asldenot{q}$ and we can prove $\angletriple{p}{\regr_1; \regr_2}{q}$ using \lrule{seq} and the two proofs given by the inductive hypothesis.

	\proofcase{$\regr = \regr_1 \regplus \regr_2$}
	For $i= 1, 2$, by inductive hypothesis on $\regr_i$ we know that there exists an assertion $p_i \in \Asl$ such that $\asldenot{p_i} = \bwsem{\regr_i} \asldenot{q}$ and $\angletriple{p_i}{\regr_i}{q}$ is provable. Therefore $\asldenot{p_1 \lor p_2} = \asldenot{p_1} \cup \asldenot{p_2} = \bwsem{\regr_1} \asldenot{q} \cup \bwsem{\regr_2} \asldenot{q} = \bwsem{\regr_1 \regplus \regr_2} \asldenot{q}$ and we can prove $\angletriple{p_1 \lor p_2}{\regr_1 \regplus \regr_2}{q}$ using \lrule{choice} and the two proofs given by the inductive hypothesis.

	Now take any $p, q \in \Asl$ such that $\bwsem{\regr} \asldenot{q} \supseteq \asldenot{p}$. By the proof above we know that there exists $p'$ such that $\bwsem{\regr} \asldenot{q} = \asldenot{p'}$ and $\angletriple{p'}{\regr}{q}$ is provable. Since $\asldenot{p} \subseteq \asldenot{p'}$, the implication $p \implies p'$ holds. Using the oracle for this implication we can prove the triple $\angletriple{p}{\regr}{q}$ using \lrule{cons} and the proof of $\angletriple{p'}{\regr}{q}$.
\end{proof}

\end{document}